\documentclass[11pt]{article}

\usepackage{geometry,setspace}
\small\normalsize
\usepackage{pdfsync}
\usepackage{hyperref}
\hypersetup{
    colorlinks=true,
    citecolor=teal,
    linkcolor=teal
}
\usepackage{caption}
\usepackage{enumerate}
\usepackage{graphicx}
\usepackage{rotating,multirow}
\usepackage[table]{xcolor}
 \usepackage{amsmath,amsfonts,amssymb,mathrsfs,amsthm}
 \usepackage{mathtools}
\usepackage{bbm}
\usepackage{bm}
\usepackage[authoryear, round]{natbib}
\usepackage{booktabs}
\usepackage{pdflscape}
\usepackage{todonotes}
\usepackage{authblk}
\usepackage{tikz-cd}
\usepackage{framed}
\usepackage{siunitx}
\usepackage{enumitem}
\usepackage{theoremref}

\usepackage{xfrac}

\usepackage{caption}
\usepackage[T1]{fontenc}
\input{glyphtounicode}
\newtheorem{proposition}{Proposition}
\newtheorem{theorem}{Theorem}
\newtheorem{lemma}{Lemma}
\newtheorem{corollary}{Corollary}
\theoremstyle{definition}
\newtheorem{definition}{Definition}

\newtheorem{remarks}{Remarks}
\newtheorem{example}{Example}

\newcommand{\indicator}[1]{\ensuremath{I_{\{#1\}}}}

\usepackage{accents}
\newcommand{\vtrans}{\ensuremath{\mathcal{V}}}

\newcommand{\downprob}{\ensuremath{\delta}}

\renewcommand{\P}{\mathbb{P}}
\newcommand{\E}{\mathbb{E}}
\newcommand{\R}{\mathbb{R}}

\newcommand{\Z}{\mathbb{Z}}
\newcommand{\N}{\mathbb{N}}

\newcommand{\eqdis}{\stackrel{\text{\tiny d}}{=}}

\newcommand{\cov}{\operatorname{cov}}

\renewcommand{\leq}{\leqslant}
\renewcommand{\geq}{\geqslant}

\newcommand{\rd}{\mathrm{d}}

\newcommand{\volfunc}{\sigma}

\newcommand{\Rforward}{\ensuremath{R}}
\newcommand{\Rbackward}{\ensuremath{B}}
\newcommand{\newtext}[1]{\textcolor{black}{#1}}

\begin{document}
\title{GARCH copulas, v-transforms \\ and D-vines for stochastic volatility}

  \author[1]{Alexandra Dias}
  \author[1]{Jialing Han}
 \author[1]{Alexander J.\ McNeil}
 \affil[1]{The School for Business and Society, University of York}
\date{\today}
\maketitle
\begin{abstract}
  The bivariate copulas that describe the dependencies and partial
dependencies of lagged variables in strictly stationary, first-order GARCH-type processes
are investigated.  It is shown that the 
 copulas of symmetric GARCH processes are jointly symmetric but
 non-exchangeable, while
 the copulas of processes with symmetric innovation distributions and
 asymmetric leverage effects have weaker h-symmetry; copulas
 with asymmetric innovation distributions have neither
 form of symmetry. \newtext{Since the true bivariate copulas are typically inaccessible,
 due to the unknown functional forms of the marginal distributions of
 GARCH processes, a new class of approximating copulas is proposed. These rely on copula
 density constructions
 that combine standard bivariate copula densities for positive dependence
 with two uniformity-preserving transformations known
 as v-transforms.
 The construction is shown to be particularly effective
       when applied to the density of the copula of the absolute values of a
       spherical t distribution.
 Tractable simplified D-vines incorporating
 the new pair copulas are developed for applications to time series
 showing stochastic volatility. The resulting models are shown
 to provide better fits to simulated data from GARCH processes,
 and to a dataset of financial exchange-rate returns,
 than have previously been obtained using vine copulas.}
\end{abstract}

\section{Introduction}\label{sec:introduction}

\newtext{This paper addresses the construction of
copula-based  models for time series displaying stochastic volatility,
such as financial asset-price return series. The philosophy of the
copula-modelling approach is that, by separating the issues of
modelling marginal behaviour and dependence structure, we may obtain
composite models that provide superior fit and predictive performance
in practice. A number of authors~\citep{bib:loaiza-maya-et-al-18,bib:bladt-mcneil-21,bib:zhao-shi-zhang-22} have suggested that
simplified D-vines, a class of models based on pair or bivariate
copulas, offer a promising framework for modelling volatile return series.}

\newtext{Finding suitable pair copulas to describe volatile return series is not
straightforward, since most existing bivariate copulas are for
modelling monotonic dependence relationships. The serial dependencies in
return series tend to be monotonic in absolute values, but
non-monotonic in the raw returns. In this paper we attempt to learn
from the success of GARCH-type models in practical applications by
studying their implied pair copulas.
Primarily, this is a paper about the
dependence structure of GARCH processes rather than empirical data, although we provide an example to
show that the insights we gain can lead to highly effective D-vine
models for real data.}

The ARCH~\citep{bib:engle-82} and GARCH~\citep{bib:bollerslev-86} prototypes
have inspired a vast class of related GARCH-type models, including the GJR-GARCH model,
which incorporates leverage effects to reflect
different contributions to volatility by returns of differing
signs~\citep{bib:glosten-jagannathan-runkle-93}, power-GARCH
models~\citep{bib:ding-engle-granger-93} and exponential GARCH
models~\citep{bib:nelson-91}, to name but a few of the more
important extensions.
Despite their success in modelling financial data, GARCH-type models are something of a black box. The
stationary distribution is not generally known in closed form,
although many models give rise to distributions with
Pareto tails~\citep{bib:mikosch-starica-00}. In particular, the serial dependence structure is not well
understood. If $(X_t)_{t\in\Z}$ follows a stationary GARCH-type model
then the distribution of $(X_1,\ldots,X_n)^\top$ for $n \geq 2$ must
have a copula, but this copula has no simple, accessible
closed form.

Vine copulas also have a large literature,
including~\citet{bib:joe-96}, \citet{bib:bedford-cooke-02},
 \citet{bib:kurowicka-cooke-06},
 \citet{bib:aas-czado-frigessi-bakken-09} and \citet{bib:smith-min-almeida-czado-10}. The
 most suitable vine structure for a time series is the D-vine, which is able
to describe strict stationarity of a random vector under 
additional translation-invariance restrictions on the vine
structure~\citep{bib:nagler-kruger-min-20}. Markov copula models~\citep{bib:darsow-nguyen-olsen-92,bib:chen-fan-06,bib:beare-10,bib:domma-giordano-perri-09} are
examples of stationary first-order D-vines and a number of
authors have written on their higher-order and multivariate
extensions~\citep{bib:ibragimov-09,bib:beare-seo-15,bib:brechmann-christian-czado-15,bib:loaiza-maya-et-al-18,bib:nagler-kruger-min-20,bib:zhao-shi-zhang-22}.

\newtext{Our goals in the present paper are threefold. First, we want
  to analyse the structure of the implied pair copulas in
first-order GARCH-type models and
to investigate how different model features affect the
symmetries of these pair copulas. Second, since the true GARCH pair copulas are
inaccessible for practical purposes, we want to find tractable copulas that can approximate
them. Finally, with the help of the analyses in the paper, we aim to propose a class of D-vines that behave like
GARCH-type models and give a first glimpse of its suitability
for real data applications. }

\newtext{Our approach is similar to the one followed in the insightful
  paper of~\cite{bib:loaiza-maya-et-al-18}, in which the authors use simulated
  data to reveal the cross-shaped dependence structure of the copula
  that describes the serial dependence structure of a classical symmetric
  ARCH(1) process\footnote{\cite{bib:loaiza-maya-et-al-18}
    and~\cite{bib:smith-maneesoonthorn-18} also analyse the
    first-order copulas of discrete-time stochastic volatility models, an
    alternative model
    class for financial asset-price returns we do not consider; they find similar
    cross-shaped behaviour.}. They use this analysis to propose D-vines of low
  Markov order (1 or 5) in which the pair copulas are mixtures of positive dependence and negative
dependence copulas; they show that these models give better forecasts
of the
conditional quantiles (value-at-risk) of exchange-rate returns than classical symmetric models of
first-order GARCH type.
Another paper related to our own is by
\cite{bib:zhao-shi-zhang-22}, who argue that t copulas in D-vines
can approximate the dependence structure of GARCH(1,1) and GJR-GARCH(1,1)
models. Our paper also extends ideas in~\cite{bib:bladt-mcneil-21} by
proposing a class of models that partly overlaps with the vt-D-vines
in that paper but offers additional flexibility.}


\newtext{To explain our contributions, we note that the joint density of the distribution of $(X_1,\ldots,X_n)^\top$ in a strictly
stationary process $(X_t)_{t\in \Z}$, such as a stationary model of GARCH type, can be
decomposed as a D-vine of the form}
                                 \begin{multline}\label{eq:27}
                     f_{X_1,\ldots,X_n}(x_1,\ldots,x_n) = 
        \prod_{i=1}^n
                  f_{X}(x_i)    \\ \prod_{k=1}^{n-1}\prod_{j=k+1}^{n}
                  c_{k}(F_{X_{1}|\bm{X}_{[2:k]}}(x_{j-k}\mid
                  \bm{x}_{[j-k+1 : j-1]}
                  ),  F_{X_{k+1}|\bm{X}_{[2:k]}}(x_{j} \mid   \bm{x}_{[j-k+1 : j-1]})
                  \mid    \bm{x}_{[j-k+1 : j-1]})
\end{multline}
where we use the compact notation $\bm{x}_{[l:m]} :=
(x_l,\ldots,x_m)^\top$, $f_X$ denotes the marginal density, $F_{X_{i}|\bm{X}_{[2:k]}}$
is the conditional distribution function
of $X_i$ given $X_2,\ldots,X_k$
and $c_k(u,v \mid \bm{x}_{[j-k+1:j-1]})$ is the density of the copula of the
conditional distribution of $(X_1,X_{k+1})^\top$ given 
$X_2=x_{j-k+1},\ldots,X_k=x_{j-1}$. In the case $k=1$ the
conditional distribution is simply the marginal
distribution $F_X$ and $c_1$ is simply the density of the copula of
the unconditional
distribution of $(X_1,X_2)^\top$. When the process is first-order
Markov, such as a model of ARCH(1) type, all of the terms involving $c_k(u,v\mid \cdot)$ for $k >1$
disappear.

Although the densities of GARCH-type models may be written as D-vines,
they may not in general be written as so-called simplified D-vines,
that is models in which the conditional copula densities $c_k$ do not
depend explicitly on the conditioning values $\bm{x}_{[j-k+1: j-1]}$ in
equation~(\ref{eq:27}) for $k\geq
2$. When D-vines are used in practical modelling applications they are
always
simplified D-vines. However, as pointed out by a number of
authors~\citep{bib:haff-aas-frigessi-10,bib:stoeber-joe-czado-13,bib:spanhel-kurz-19,bib:mroz-fuchs-trutschnig-21},
these may not be able to capture all the features present in a given joint density, such as that of a GARCH-type model.

 \newtext{Our first contributions relate to the analysis of the GARCH pair
   copula densities $c_k$ in~(\ref{eq:27}) for first-order GARCH-type models and are described in
Section~\ref{sec:arch-garch-copulas}.
We provide theory to show that (i)
  symmetric models give rise to jointly symmetric copulas, (ii) models
  with symmetric innovation distributions and leverage
  have copulas with weaker symmetry, which we call h-symmetry, and
  (iii) models with asymmetric innovations lose all symmetry. We confirm the theoretical insights by using numerical approximations to provide accurate
  contour plots of the copula densities which reveal hitherto overlooked
  features, such as the non-exchangeability of copulas even in the
  most symmetric models. The theory complements and
  extends~\cite{bib:loaiza-maya-et-al-18}, who consider the 
  copula density $c_1$ in the ARCH(1) model
  and use a simulation approach and adaptive kernel density estimation
  to elicit its cross-shaped
  form.}

\newtext{The second body of material, found in
  Section~\ref{sec:copula-dens-constr} relates to the search for copulas
  that can approximate the GARCH pair copulas.  We show how uniformity-preserving transformations known as
     v-transforms~\citep{bib:mcneil-20} play an implicit role in the
     copulas of classic symmetric models and in models with certain
     forms of asymmetry. We investigate methods of constructing
     copulas to approximate GARCH pair copulas using a
     technique we call stochastic inversion of v-transforms and we identify the copulas
     that provide the best fit to simulated data
     from GARCH-type models. We find that inverse-v-transformed
     versions of positive dependence copulas with upper tail
     dependence, such as survival Clayton, Joe, or a new copula we
     refer to as absolute spherical t, typically fit better than
     mixture copulas.}

   \newtext{Finally, in Section~\ref{sec:simplified-D-vines}, we propose a class of
     simplified D-vines that behave like GARCH models and that could
     be used to model real data. To support the use of
     simplified D-vines, we provide some evidence in
     Section~\ref{sec:high-order-cond} that the
     dependence of higher order conditional copulas on conditioning
     variables may be weak. Our simplified D-vines combine the inverse-v-transformed pair copulas
     developed in Section~\ref{sec:copula-dens-constr} with a parsimonious method
     of parameterization via the partial autocorrelation
     function of ARMA processes as proposed by~\cite{bib:bladt-dias-han-mcneil-25}. The
     resulting models generalize the vt-D-vine models of~\cite{bib:bladt-mcneil-21} and provide a better fit to simulated GARCH data than the
     simplified D-vines of order 1 and 2 based on t copulas used
     in~\cite{bib:zhao-shi-zhang-22}. A first application 
     to a real-world dataset of daily exchange-rate returns
     provided and analysed by~\cite{bib:loaiza-maya-et-al-18} suggests
     that the much high Markov order of our models, and the copula
     constructions they contain, can lead to better-fitting models than have
     previously been obtained using D-vines.}

\section{An analysis of GARCH pair
  copulas}\label{sec:arch-garch-copulas}

\subsection{ARCH and GARCH processes}

We define a broad class of 
first-order GARCH-type models and establish nomenclature for some
special cases.

\begin{definition}[First-order GARCH processes]\label{def:arch1-copula}
  \begin{enumerate}
  \item A \textit{process of GARCH(1,1) type} is a strictly stationary process $(X_t)_{t\in\Z}$ satisfying
equations of the form
\begin{equation}
  \label{eq:1b}
  X_t = \sigma_t\; \epsilon_t
\end{equation}
 where $(\epsilon_t)_{t\in\Z}$ is a sequence of iid innovations, drawn
 from a continuous distribution with
 mean zero and variance one, and $(\sigma_t)_{t\in\Z}$ is a sequence of
 variables satisfying $\sigma_t = \volfunc(X_{t-1},
   \sigma_{t-1})$ where
$\volfunc: \R \times \R^+ \to \R^+$ is a strictly positive-valued
parametric volatility function.
\item A \textit{process of ARCH(1) type} is a process of GARCH(1,1) type where
 $\volfunc(x,s)$ depends only on
  $x$.
 \item 
The \textit{classical GARCH(1,1) process} is a process
of GARCH(1,1) 
type in which $\volfunc(x,s)$ satisfies
\begin{equation}
  \label{eq:1}
  \volfunc(x, s) = \sqrt{\alpha_0 + \alpha_1 x^2 + \beta_1
  s^2}
\end{equation}
where $\alpha_0>0$, $\alpha_1 \geq 0$, $\beta_1 \geq 0$ and 
$\E( \ln ( \alpha_1 \epsilon_1^2 +\beta_1)) < 0$ for an innovation
$\epsilon_1$.
\item The \textit{classical ARCH(1) process} is the classical
  GARCH(1,1) process with $\beta_1 = 0$.
\item A process of GARCH(1,1) type is said to have \textit{symmetric
    distribution}, if the distribution of $\epsilon_1$ has a density
  that is symmetric around zero, and to have \textit{symmetric
    volatility} if $\volfunc(x,s)$ is an even function of
  $x$. It is said to be \textit{symmetric} if
    it has symmetric distribution and symmetric volatility.
 \end{enumerate}
\end{definition}

The mathematical properties of the classical GARCH(1,1) process
have been widely studied and the conditions
on the parameters in part 3 of Definition~\ref{def:arch1-copula} are known to guarantee a strictly
 stationary solution to the equations~(\ref{eq:1b})~\citep{bib:nelson-90a,bib:bougerol-picard-92}. For a finite-variance solution it is
 usual to assume the stronger condition that $\alpha_1 + \beta_1 < 1$.
 \newtext{It is known that the
 marginal or stationary distribution $F_X$ has regularly varying
 (power-law) tails
 with a tail index $\zeta$ that
 may be calculated as the root of the equation
 $\E\left(  \left( \alpha_1\epsilon_1^2 +
   \beta_1 \right) ^{\zeta/2}  \right)=
 1$~\citep{bib:mikosch-starica-00}. However, the exact functional forms of $F_X$ and its density $f_X$
 are not available.}

 \subsection{Symmetry concepts for bivariate distributions}
 \newtext{In our analysis of the pair copulas of GARCH processes, we first
   consider models with symmetric innovation
   distributions. To describe the resulting copulas, it helps to
   recall the definitions of radial and joint symmetry
   in~\cite{bib:nelsen-99} and to identify two further concepts we call
   h- and v-symmetry (symmetry about horizontal and vertical axes).}
 \begin{definition}\label{def:symmetry}
   Let $(Y,Z)^\top$ be a random vector taking values in a set $A
   \subseteq \R^2$.
   \begin{enumerate}[label={(\roman*)},itemindent=1em]
   \item If $(Y,Z-b)^\top \eqdis (Y,b-Z)^\top$  for some $b\in\R$ we
     say that $(Y,Z)^\top$ 
     is h-symmetric about $b$.
            \item If $(Y-a,Z)^\top \eqdis (a-Y,Z)^\top$  for some
              $a\in\R$ we say that $(Y,Z)^\top$ 
         is v-symmetric about $a$.
         \item If $(Y-a,Z-b)^\top \eqdis (a-Y,b-Z)^\top$  for
           $a,b\in\R$ we say that $(Y,Z)^\top$
           is radially symmetric
           about $(a,b)$.
 \item If (i), (ii) and (iii) all hold, we say that $(Y,Z)^\top$ is jointly symmetric about $(a,b)$.
   \end{enumerate}
 \end{definition}
 \begin{remarks}
   \begin{enumerate}
     \item We also describe the joint distribution function (df)
       $F_{YZ}$ of $(Y,Z)^\top$ as having these symmetries.
      \item If (i) holds, the marginal distribution of $Z$ is symmetric
     about $b$; if (ii) holds, the marginal distribution of $Y$ is symmetric
     about $a$. The reverse implications are not true.
       \item 
   H-symmetry and v-symmetry together imply
 radial and joint symmetry. However radial
 symmetry does not necessarily imply h-symmetry and v-symmetry,
 although it does imply symmetry of both marginal distributions.
 \item \citet{bib:nelsen-99} notes that jointly symmetric distributions have
zero correlation (when finite second moments permit the measure to be
defined). In Section~\ref{sec:corr-rank-corr} of the Online
Supplement, we show that h-symmetry (or v-symmetry) on its own is sufficient for an
absence of both correlation and rank correlation.
Despite the absence of correlation, distributions with these forms
of symmetry can have quite strong
forms of dependence.
\end{enumerate}
\end{remarks}

\newtext{
\noindent If the distribution functions $F_Y$ and $F_Z$ of
$Y$ and $Z$ are continuous, the random vector $(U,V)^\top=
(F_Y(Y), F_Z(Z))^\top$ has distribution function $C$, known as the
copula of $(Y,Z)^\top$, and the symmetries have the following
implications for $C$.}
\begin{lemma}\label{lemma:symm-copulas}
 Let $(Y,Z)^\top$ be a random vector with continuous marginal
 distributions and copula $C$.
  \begin{enumerate}[label={(\roman*)},itemindent=1em]
  \item If $(Y,Z)^\top$ is h-symmetric then $ C(u,v) = u - C(u,1-v)$, for $ (u,v) \in
  [0,1]^2$.
  \item  If $(Y,Z)^\top$ is v-symmetric then $C(u,v) = v - C(1-u,v)$, for $ (u,v) \in
  [0,1]^2$,
  \item If $(Y,Z)^\top$ is radially symmetric then $C(u,v) = u + v +
    C(1-u,1-v) -1$ , for $ (u,v) \in
  [0,1]^2$.
  \end{enumerate}
\end{lemma}
\begin{proof}
  The proof of (iii) is given by~\citet{bib:nelsen-99} (Theorem 2.7.3)
  and (i) and (ii) follow by a similar method.
\end{proof}
\noindent \newtext{This result implies that the random vector $(U,V)^\top=
(F_Y(Y), F_Z(Z))^\top$, or its df $C$, inherits the properties of h-symmetry,
v-symmetry, radial and joint symmetry from $(Y,Z)^\top$; the corresponding
lines and centre of symmetry
are at 0.5 and $(0.5,0.5)$.}
\newtext{
The next result provides the first connection between symmetries
of copula functions
and v-transforms. The function $\vtrans(u) = |2u-1|$
appearing in this result is the 
simplest example of a v-transform.}

\begin{proposition}\label{prop:copula-absolute}
  \newtext{
  Let $C$ be the distribution function of $(U,V)^\top$. Then, for
  all $(u,v) \in [0,1]^2$,
  \begin{enumerate}
  \item[(i)]  if $C$ is h-symmetric,
  $C(u,v) = (u + (-1)^{\indicator{v < 0.5}}
  C^{(h)}(u, |2v-1|))/2$ where $C^{(h)}$ is the df of $(U,|2V-1|)^\top$;
   \item[(ii)]  if $C$ is v-symmetric,
  $C(u,v) = (v + (-1)^{\indicator{u < 0.5}}
  C^{(v)}(|2u-1|, v))/2$ where $C^{(v)}$ is the df of $(|2U-1|,V)^\top$;
\item[(iii)]  if $C$ is jointly symmetric,
 $   C(u,v) =  
  (2u+2v -1 + (-1)^{\indicator{(u-0.5)(v-0.5) < 0 }} C^*(|2u-1|, |2v-1|))/4$
  where $C^*$ is the df of $(|2U-1|,|2V-1|)^\top$.
\end{enumerate}
}
\end{proposition}
\begin{proof}
  \newtext{
  We prove (i) and (iii) since (ii) follows immediately from (i). We use
  the equalities
  $C^{(h)}(u,v) = \P(U \leq u , |2V-1| \leq
  v) = \P(U \leq u , |V-0.5| \leq v/2) = C(u,(1+v)/2)
  - C(u,(1-v)/2)$ and part (i) of
  Lemma~\ref{lemma:symm-copulas} to infer that
  $  C^{(h)}(u, |2v-1|) =
  C(u,1-v) - C(u,v) = u - 2C(u,v)$ if $v < 0.5$ and
  $C^{(h)}(u, |2v-1|) =C(u,v) - C(u,1-v) = 2C(u,v) - u$ if $v \geq
  0.5$. Hence part (i) follows.
  }

   \newtext{To prove (iii) note that, if $C$ is jointly symmetric,
     then $C^{(h)}$ is v-symmetric since
$C^{(h)}(u,v) + C^{(h)}(1-u,v) = C(u,(1+v)/2) - C(u,(1-v)/2) + C(1-u,
(1+v)/2) - C(1-u, (1-v)/2) = (1+v)/2 -(1-v)/2 =v$. Hence
    $C^{(h)}(u, |2v-1|) = (|2v-1| + (-1)^{\indicator{u < 0.5}}
  C^*(|2u-1|, |2v-1|))/2$. In combination with part (i) this yields
  part (iii).}
\end{proof}
If $C$ is the copula of the h-symmetric
random
vector $(Y,Z)^\top$ in part (i) of Lemma~\ref{lemma:symm-copulas}, 
$C^{(h)}$ in part (i) of Proposition~\ref{prop:copula-absolute}
is the copula of $(Y,|Z-b|)^\top$.
This follows by observing
that if $V = F_Z(Z)$ is the
  probability-integral transform of $Z$, then the probability-integral transform
  of $|Z-b|$ is
  $F_{|Z-b|}(|Z-b|) = |2V-1|$.
Similarly, if $C$ is the copula of
the jointly symmetric
random
vector $(Y,Z)^\top$, then $C^*$  in part (iii) of Proposition~\ref{prop:copula-absolute}
is the copula of $(|Y-a|,|Z-b|)^\top$. 

\newtext{In this paper we apply the symmetry concepts to random vectors
  $(Y,Z)^\top$ with a
  joint density $f_{Y,Z}(y,z) > 0$ on $\R^2$, implying that the copula $C$ has density
  $c$ given by
  $$
  c(u,v) = f_{Y,Z}(F_Y^{-1}(u),
  F_Z^{-1}(v))/\left( f_Y(F_Y^{-1}(u)) f_Z(F_Z^{-1}(v))\right),
  $$
where $f_Y$ and $f_Z$ denote the marginal densities.}
It is easily inferred from
 Definition~\ref{def:symmetry} that
 h-symmetry of $(Y,Z)^\top$ is equivalent to $f_{Y,Z}(y,b+z) = f_{Y,Z}(y,b-z)$ for all
 $(y,z)\in \R^2$ and some $b\in\R$, while v-symmetry
 and radial symmetry are equivalent to 
 $f_{Y,Z}(a+y,z) = f_{Y,Z}(a-y,z)$ for some $a\in\R$ and
 $f_{Y,Z}(a+y,b+z) = f_{Y,Z}(a-y,b-z)$ for some $a,b\in\R$. Similarly, 
h-symmetry of the
copula $C$ is equivalent to
the identity $c(u,v) =
  c(u,1-v)$, 
v-symmetry is equivalent to $c(u,v) = c(1-u,v)$ and radial symmetry is equivalent to $c(u,v) = c(1-u,1-v)$.
\newtext{When densities exist, the following Corollary of
  Proposition~\ref{prop:copula-absolute} is immediate.}

\begin{corollary}\label{prop:copula-absolute-cor}
  \newtext{
   Let $C$ be the distribution function of $(U,V)^\top$ and
   assume $C$ has joint density $c$. For
  all $(u,v) \in [0,1]^2$,
  \begin{enumerate}
  \item[(i)]  if $C$ is h-symmetric,
  $c(u,v) =
  c^{(h)}(u, |2v-1|)$ where $c^{(h)}$ is the density of 
  $(U,|2V-1|)^\top$;
   \item[(ii)]  if $C$ is  jointly symmetric,
  $c(u,v) =
  c^*(|2u-1|, |2v-1|)$ where $c^*$ is the density of 
  $(|2U-1|,|2V-1|)^\top$.\end{enumerate}}
\end{corollary}
\noindent \newtext{Simple calculations show that part (ii) also implies the following series of
  equations which will be used in the sequel:
  \begin{equation}
    \label{eq:8}
    c^*(u,v) = c\left( \frac{1-u}{2},\frac{1-v}{2}\right) = c\left(
      \frac{1+u}{2},\frac{1-v}{2}\right) =
    c\left( \frac{1-u}{2},\frac{1+v}{2}\right) =
    c\left( \frac{1+u}{2},\frac{1+v}{2}\right).
  \end{equation}
}

 \subsection{The first-order copula of GARCH(1,1)-type processes}

\newtext{We can now apply these symmetry ideas to GARCH-type processes.}  Simple calculations show that the joint density of a pair of successive variables $(X_1,X_2)^\top$
  in a process of GARCH(1,1) type is given by
\begin{align}
  f_{X_{1} ,X_{2}}(x, y) = \int_0^\infty  f_{X_{1} ,X_{2},
    \sigma_1}(x, y, s) \rd s &= \int_0^\infty
f_{X_2 \mid X_1,\sigma_{1}}(y \mid x, s) f_{X_1 \mid \sigma_1}(x
\mid s ) f_\sigma(s)
\rd s \nonumber \\
&= \int_0^\infty \frac{1}{\volfunc(x,s)} f_\epsilon\left(
     \frac{y}{\volfunc(x,s)} \right)
 \frac{1}{s}
      f_{\epsilon}\left(\frac{x}{s}\right) f_\sigma(s) \rd s \label{eq:garch-density}
\end{align}
where $f_\sigma$ is the marginal density of the
volatility process $(\sigma_t)_{t\in\Z}$ and $f_\epsilon$ is the
density of the innovations. The 
density $c_1$ of the copula $C_1$ of $(X_1,X_2)^\top$ is thus
\begin{equation}
  \label{eq:19}
  c_1(u,v) = \frac{
    \int_0^\infty \frac{1}{s\volfunc\left(F_X^{-1}(u),s\right)} f_\epsilon\left(
     \frac{F_X^{-1}(v)}{\volfunc\left(F_X^{-1}(u),s\right)} \right)
      f_{\epsilon}\left(\frac{F_X^{-1}(u)}{s}\right) f_\sigma(s) \rd s
  }{ f_X\left(F_X^{-1}(u)\right) f_X\left(F_X^{-1}(v)\right)}
\end{equation}
where 
\begin{equation}\label{eq:32}
   f_{X}(x) = \int_{0}^\infty f_{X_1 \mid \sigma_{1}}(x \mid s)
                 f_{\sigma}(s) \rd s = \int_{0}^\infty \frac{1}{s} f_\epsilon\left(\frac{x}{s}\right)
 f_\sigma(s) \rd s
\end{equation}
is the marginal density of the process $(X_t)_{t\in\Z}$.

In the case of a process of
ARCH(1) type (where $\volfunc(x,s)$ only depends on $x$) these expressions
simplify considerably and become
\begin{align}
  f_{X_{1} ,X_{2}}(x, y) & =\frac{1}{\volfunc(x)} f_\epsilon\left(
     \frac{y}{\volfunc(x)} \right) f_X(x),\quad\quad c_1(u,v) = \frac{1}{\volfunc\left(F_X^{-1}(u)\right)}f_\epsilon \left(
     \frac{F_X^{-1}(v)}{\volfunc\left(F_X^{-1}(u)\right)} \right)
   \frac{1}{f_X(F_X^{-1}(v))} \,.\label{eq:22}
\end{align}

The key insight from these equations is the following.
\begin{proposition}\label{prop:symmetry}
\begin{enumerate}[label={(\roman*)},itemindent=1em]
\item In a GARCH(1,1)-type process with symmetric distribution, the
 distribution of $(X_1,X_2)^\top$ is h-symmetric about zero.\label{thm:symmetry1}
\item In a symmetric GARCH(1,1)-type process, the distribution of
  $(X_1,X_2)^\top$ is jointly symmetric about the origin.\label{thm:symmetry2}
\end{enumerate}
\end{proposition}
\begin{proof}
  We consider the joint density $f_{X_1,X_2}(x,y)$ in~\eqref{eq:garch-density}.
\begin{enumerate}[label={(\roman*)},itemindent=1em]
\item It is clear that a symmetric innovation distribution implies
  that $f_{X_1,X_2}(x,y) = f_{X_1,X_2}(x,-y)$, meaning the
 random vector
  $(X_1,X_2)^\top$ is h-symmetric about zero. 
  \item If the volatility function $\sigma(x,s)$ is an even function
    of $x$, this also implies that
    $f_{X_1,X_2}(x,y) = f_{X_1,X_2}(-x,y)$ and so $(X_1,X_2)^\top$
  is  v-symmetric about zero and hence radially and jointly symmetric about
  the origin. 
  \end{enumerate}
\end{proof}
In case (i) above we can apply Lemma~\ref{lemma:symm-copulas} to infer that the 
copula $C_1$ is h-symmetric and we can apply
Proposition~\ref{prop:copula-absolute} and Corollary~\ref{prop:copula-absolute-cor}  to infer that its density may be written in the
form $c_1(u,v) = c_{1}^{(h)}(u, |2v-1|)$ where $c_{1}^{(h)}$ is the density of the
copula of $(X_{1}, |X_2|)^\top$. In case (ii) we infer that the
copula $C_1$ is jointly symmetric and  its density may be written in the
form $c_1(u,v) = c_{1}^*(|2u-1|, |2v-1|)$ where $c_{1}^*$ is the density of the
copula of $(|X_{1}|, |X_2|)^\top$. In the latter case it follows from~\eqref{eq:8} that
$c_{1}^*(u, v) = c_1((u+1)/2, (v+1)/2)$, which means that the
copula density $c_{1}^*$ may be viewed as a blown-up version
of the upper quadrant of the copula density $c_1$, i.e.\ the part of
the density on $[0.5,1]^2$.

Despite the strong symmetry revealed by
Proposition~\ref{prop:symmetry}, the bivariate distribution of
$(X_1,X_2)^\top$ in a symmetric process of GARCH(1,1) type is not in general exchangeable. This implies that ARCH and
GARCH are not 
time-reversible models; in general 
$f_{X_2\mid X_{1}}(y \mid x) \neq f_{X_{1}\mid X_{2}}(y \mid x)$ and thus,
in an intuitive sense, they are models in which the predictive
distribution of the future given the past differs from the
``predictive distribution'' of the past given the future. From an
econometric point of view, this would be considered natural.
The non-exchangeability carries over to the copula
so that we have in general that $C_1(u,v) \neq C_1(v,u)$ on
$[0,1]^2$.

\subsection{Computational examples}\label{sec:comp-exampl}
In this section we use computational methods to approximate and plot
the copula density $c_1$ in a number of
models of GARCH(1,1) type; in Section~\ref{sec:pictures} of the Online
Supplement we give some
accompanying pictures of the joint density $f_{X_1, X_2}$ for selected
models.

We concentrate on
processes of ARCH(1) type in our illustrations since their first-order Markov structure
simplifies computation of the unknown marginal density $f_X$ and
distribution function $F_X$, which are key ingredients
in~\eqref{eq:22}. Beginning with the equation
\begin{equation}
   \label{eq:3}
   f_X(x) = \int_{-\infty}^\infty   f_\epsilon \left(
     \frac{x}{\volfunc(y)} \right)
   \frac{1}{\volfunc(y)} f_X(y)\rd y,
 \end{equation}
 which follows easily from~\eqref{eq:22}, we approximate the
integral by a sum of weighted values of $f_X(y_j)$ on a grid of points
$\mathcal{Y} = \{y_1, \ldots, y_n\}$; we can then solve a linear
system for $f_X(y_j), y_j \in \mathcal{Y}$, which is accurate enough for
our purposes.
 
 We also
consider the classical GARCH(1,1) model with Gaussian
innovations. To analyse this process we use simulation and estimation
to find a parametric model for
the unknown marginal density $f_\sigma$ of the volatility process. We
can then integrate out volatility in expressions
like~(\ref{eq:garch-density}),~(\ref{eq:19}) and~(\ref{eq:32})\footnote{All details
  of the model for $f_\sigma$ and the discretization method for
  ARCH(1)-type processes can be found in the code that
  accompanies this paper at \ttfamily{https://github.com/ajmcneil/papers}.}.
The parameter values for all the models we consider are summarized in
Table~\ref{tab:parameters}. \newtext{In the examples that follow
  we concentrate
on illustrating the copulas densities $c_1$; corresponding
illustrations of the joint densities $f_{X_1,X_2}$ can be found in Section~\ref{sec:pictures} of the Supplementary Material.}

\begin{table}[htb]
  \centering
  \begin{tabular}{lrrrr} \toprule
    Process & $\alpha_0$ & $\alpha_1$ & $\beta_1$ & $\gamma_1$ \\ \midrule
    ARCH(1) & 0.4 & 0.6 & 0 & 0 \\
    ARCH(1) with leverage effect & 0.4 & 0.3 & 0 & 0.4 \\
    GARCH(1,1) & 0.1 & 0.3 & 0.6 & 0 \\
    GARCH(1,1) with leverage effect & 0.1 & 0.1 & 0.6 & 0.4 \\ \bottomrule
  \end{tabular}
  \caption{Parameter values for processes analysed in
    paper. The volatility function $\sigma(x,s)$ in all
    models satisfies $\sigma(x,s)^2 = \alpha_1 + (\alpha_1 + \gamma_1 \indicator{x <
      0})x^2 +\beta_1 s^2$. Innovation distributions used
  are either Gaussian, Student t with $\xi=4$ or $\xi=2.5$ degrees of
  freedom or skewed Student t with $\xi=4$ degrees of freedom and
  skewness parameter $\lambda = 0.8$; all innovation distributions are scaled to
  have mean zero and variance one.\label{tab:parameters}}
\end{table}

\begin{figure}[!ht]
\centering
\includegraphics[width=7cm,height=7cm]{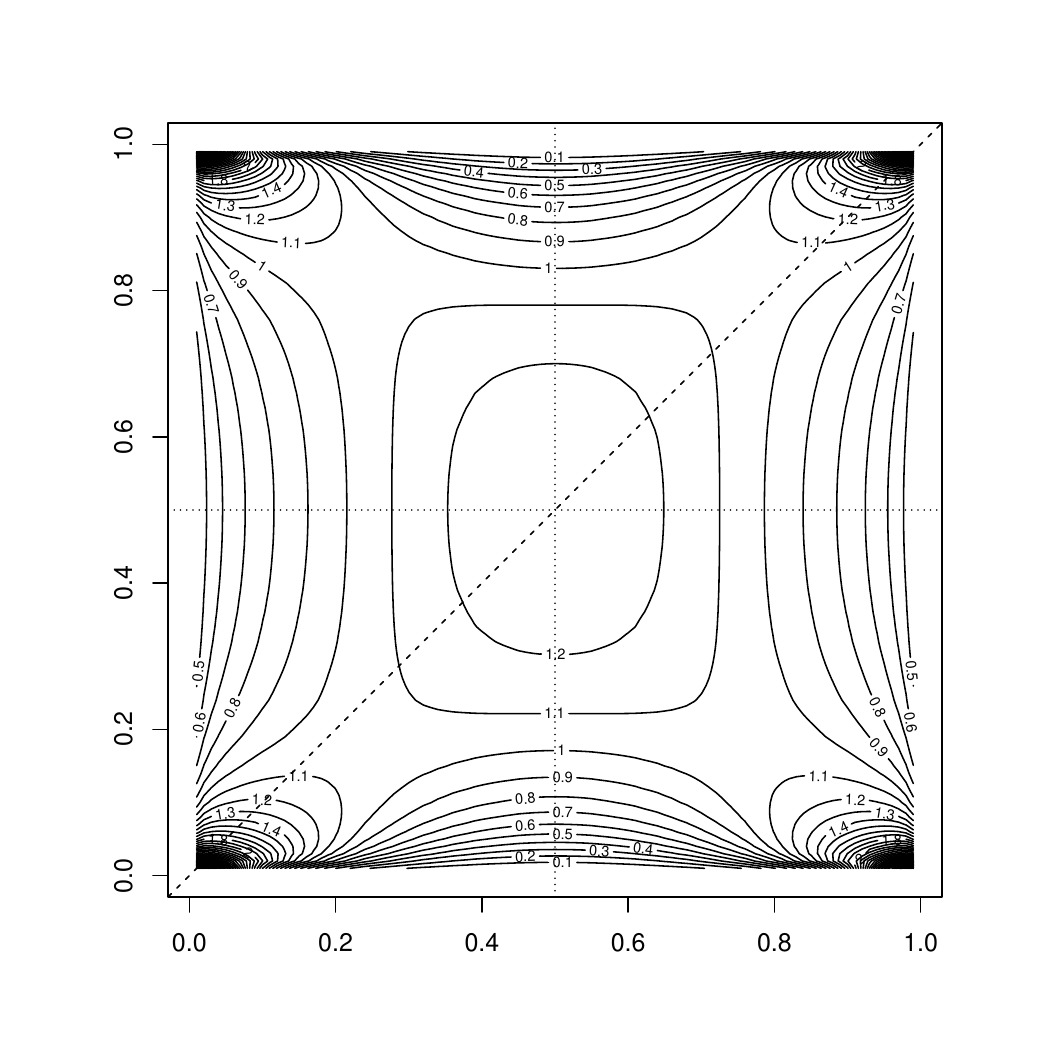}
\includegraphics[width=7cm,height=7cm]{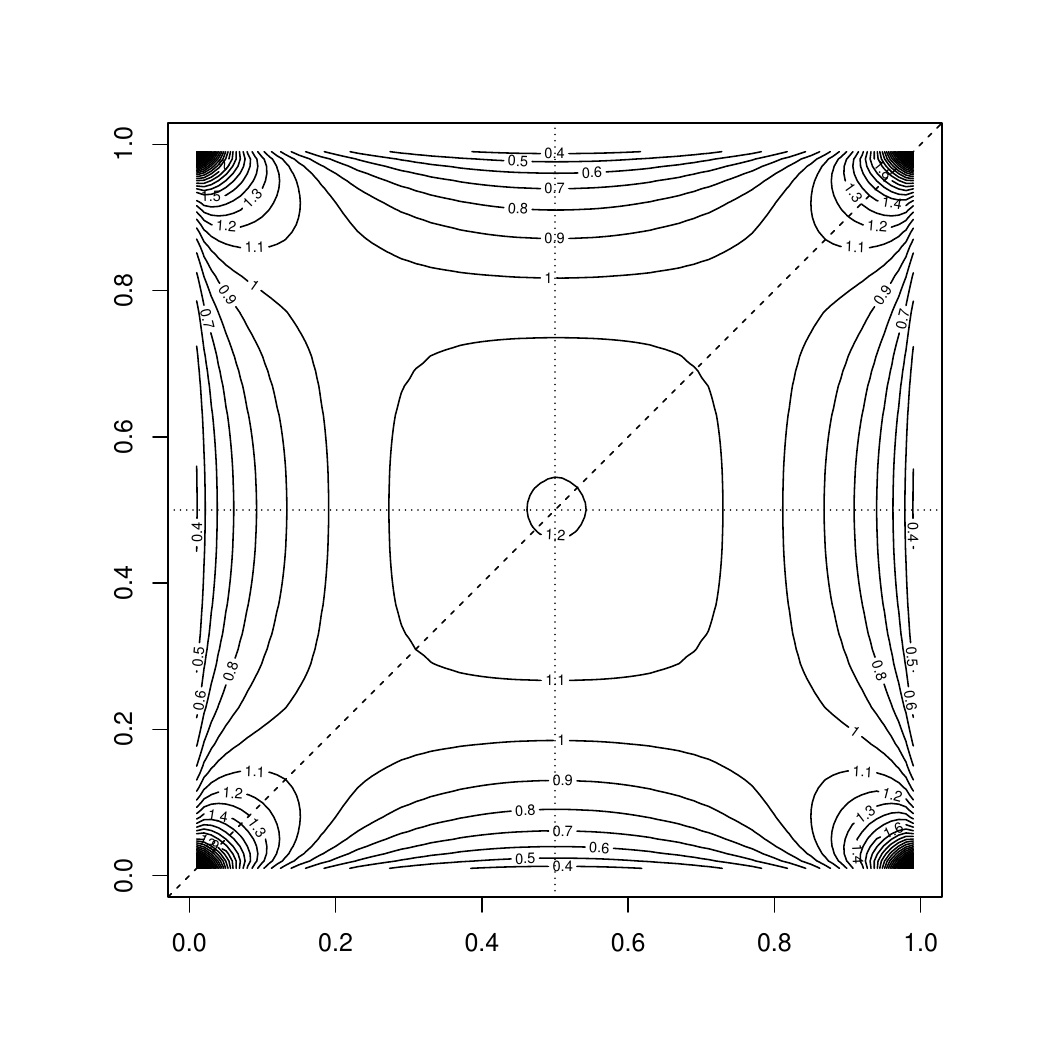}\\
\includegraphics[width=7cm,height=7cm]{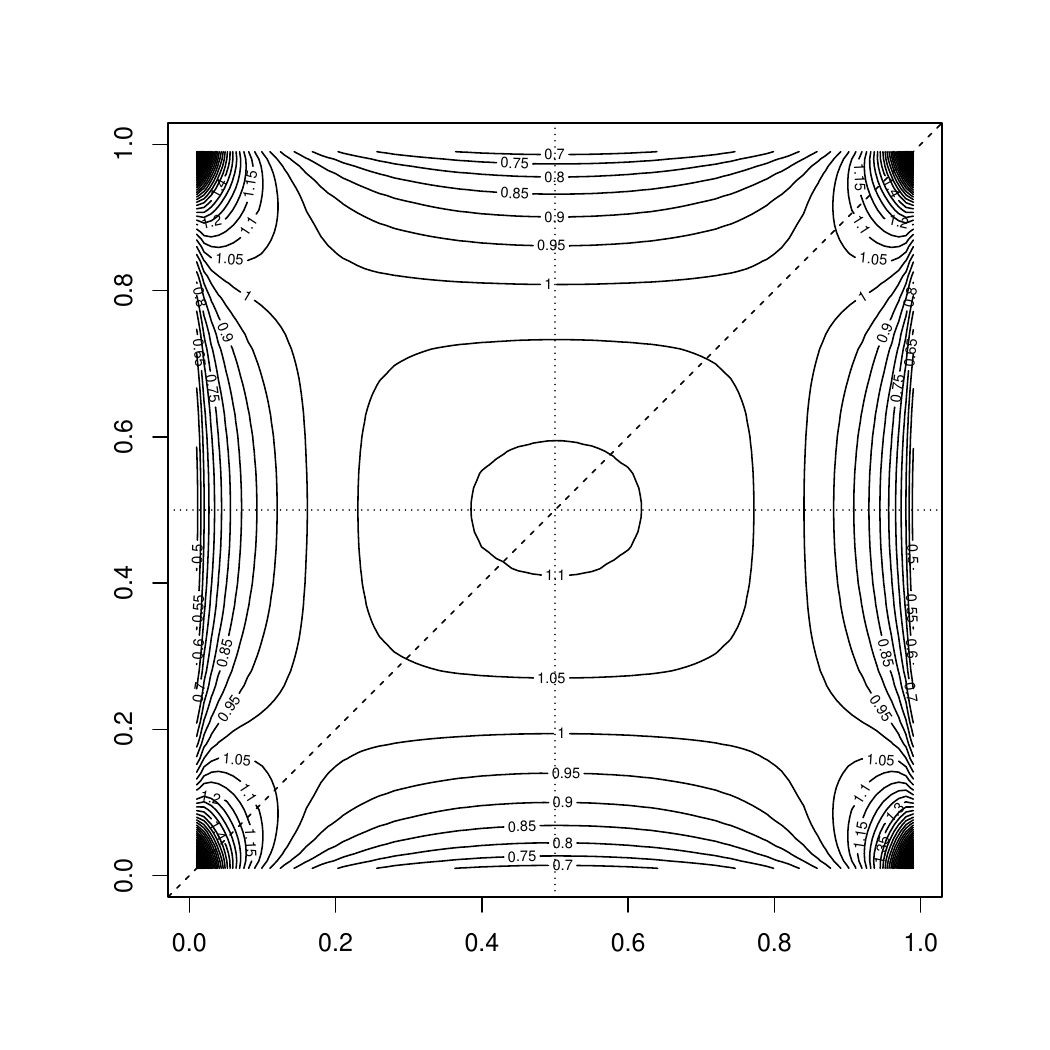}
\includegraphics[width=7cm,height=7cm]{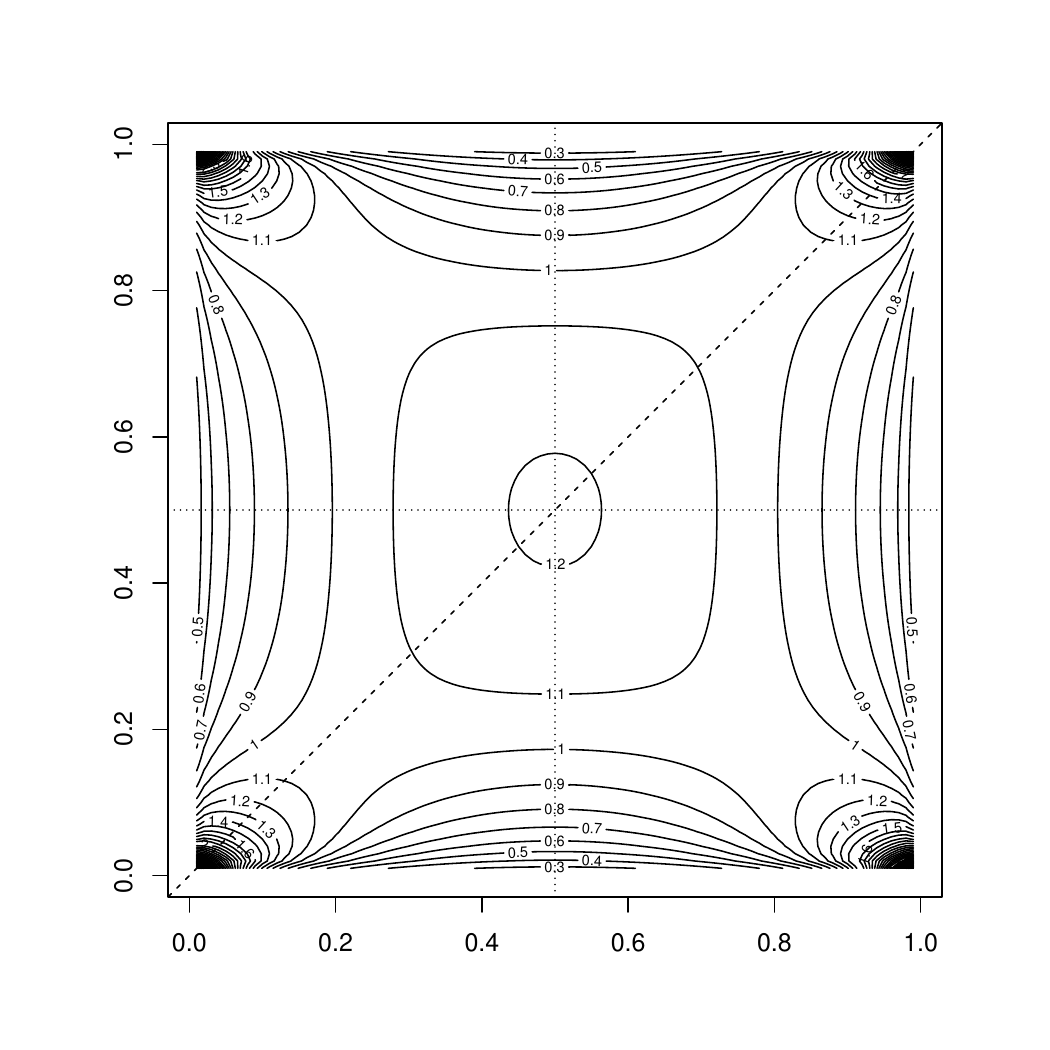}
 \caption{Contour plots of $c_1(u,v)$ for symmetric GARCH(1,1)-type processes; see Table~\ref{tab:parameters} for all
   parameter values. Top left, top right, bottom left: classical ARCH(1)
   processes with,
   respectively, Gaussian, Student t4 and Student t2.5 innovations. Bottom right: classical
   GARCH(1,1) process with Gaussian innovations.
}
\label{fig:1}
\end{figure}

 \subsubsection{Classical ARCH(1) process with Gaussian innovations (symmetric)}\label{sec:class-arch1-plot}
 For the parameter values in Table~\ref{tab:parameters}, the
  contours of the copula density $c_1$ are shown in the top left
  panel of Figure~\ref{fig:1}
The
  non-exchangeability, evidenced by a lack of symmetry in the dashed line
  $v=u$, is evident. We can also observe the
  joint symmetry, that is the h-symmetry and v-symmetry also implying
  rotational symmetry of order 2.
  Changing $\alpha_1$ would change the degree of
  dependence; in particular,
  increasing $\alpha_1$ 
  would strengthen the cross-shaped character of the copula density.

  In Section~\ref{sec:pictures} of the Supplementary Material we also
  include a contour plot of the absolute value copula for this model,
  that is the copula of $(|X_1|,|X_2|)^\top$. The
  picture for the absolute value copula is confirmed to be a blown-up version of
  the top right quadrant of the contour plot of the copula density
  $c_1$.

\subsubsection{Classical ARCH(1) process with Student innovations (symmetric)} This is
  shown in the top right panel of Figure~\ref{fig:1} for a process with the same parameters
  as in the previous example and degree of freedom $\xi =4$ for the
  innovation distribution. Remarkably, the copula appears to be very close to an
  exchangeable copula for this combination of parameters; the joint
  symmetry is again apparent. The
  bottom left panel of Figure~\ref{fig:1} (see also Figure~\ref{fig:0}) show a process in which $\xi
  =2.5$ which shows non-exchangeability
  more clearly.

  It is notable that, although the heaviness of the tail of the marginal
  distribution increases with decreasing degree of freedom, the dependence
  actually weakens. We can estimate a
  distance $D_1 = \int_0^1\int_0^1 | c_1(u,v) -1| \rd u \rd v$ between the copula density $c_1$ and the density of
  the independence copula; for the Gaussian, t4 and t2.5 copulas we obtain
  respectively 0.213, 0.158 and 0.092.

  \subsubsection{Classical GARCH(1,1) process with Gaussian innovations
    (symmetric)}\label{sec:class-garch11-plot}
  This is shown in the bottom right panel of
    Figure~\ref{fig:1} for a process with Gaussian
    innovations.
   It is well known that, for a covariance stationary
   GARCH(1,1) process ($\alpha_1 + \beta_1 <1$), the squared process
   $(X_t^2)_{t\in \Z}$ behaves like an ARMA(1,1) process with
   autoregressive parameter $\alpha_1+\beta_1$ and moving-average
   parameter $-\beta_1$; see, for
   example,~\citet[Section 4.2.2]{bib:mcneil-frey-embrechts-15}. Thus the value of
   $\alpha_1 + \beta_1$ controls the persistence, or the rate at which
   serial dependencies in the squared (or absolute value) process die
   away with lag. For a fixed value of $\alpha_1 + \beta_1$ the magnitude of the serial
   dependency at lag 1 increases with $\alpha_1$, so that the ARCH(1)
   special case actually has the strongest dependence and the copula
   tends towards independence as $\alpha_1 \to 0$. The parameters we
   have chosen are not necessarily typical of real financial data, where
   the ARCH effect $\alpha_1$ might be smaller and the GARCH
   effect $\beta_1$ larger, but have instead been selected to give a picture that has
   more evident dependence. The main point is to show that the
   behaviour of the GARCH(1,1) copula is qualitatively similar to that
   of the ARCH(1) copula.
    
\subsubsection{ARCH(1) model with Gaussian innovations and leverage
  (symmetric distribution)} In the left panel of
  Figure~\ref{fig:2} we show the copula $c_1$ for a model with
  symmetric innovations and
  volatility function $\volfunc(x)$ satisfying $\volfunc(x)^2 =
  \alpha_0 + (\alpha_1 + \gamma_1 \indicator{x <
      0})x^2$ for $\gamma_1 > 0$, so that a negative return 
  produces a larger value of the volatility than a positive return of
  the same magnitude. 
  It can be seen that leverage removes
  radial symmetry and
  v-symmetry, but there is still h-symmetry (symmetry in the line $v=
  0.5$).

\subsubsection{Classical ARCH(1) model with skewed Student innovations
  (symmetric volatility)}  The copula density $c_1$ is
  shown in the right panel of Figure~\ref{fig:2} for a model with an innovation
  distribution that is skewed to the left (towards negative
  shocks). The skewness removes all the previous symmetries in the
  pictures.

  \begin{figure}[!ht]
  \centering
  \includegraphics[width=7cm,height=7cm]{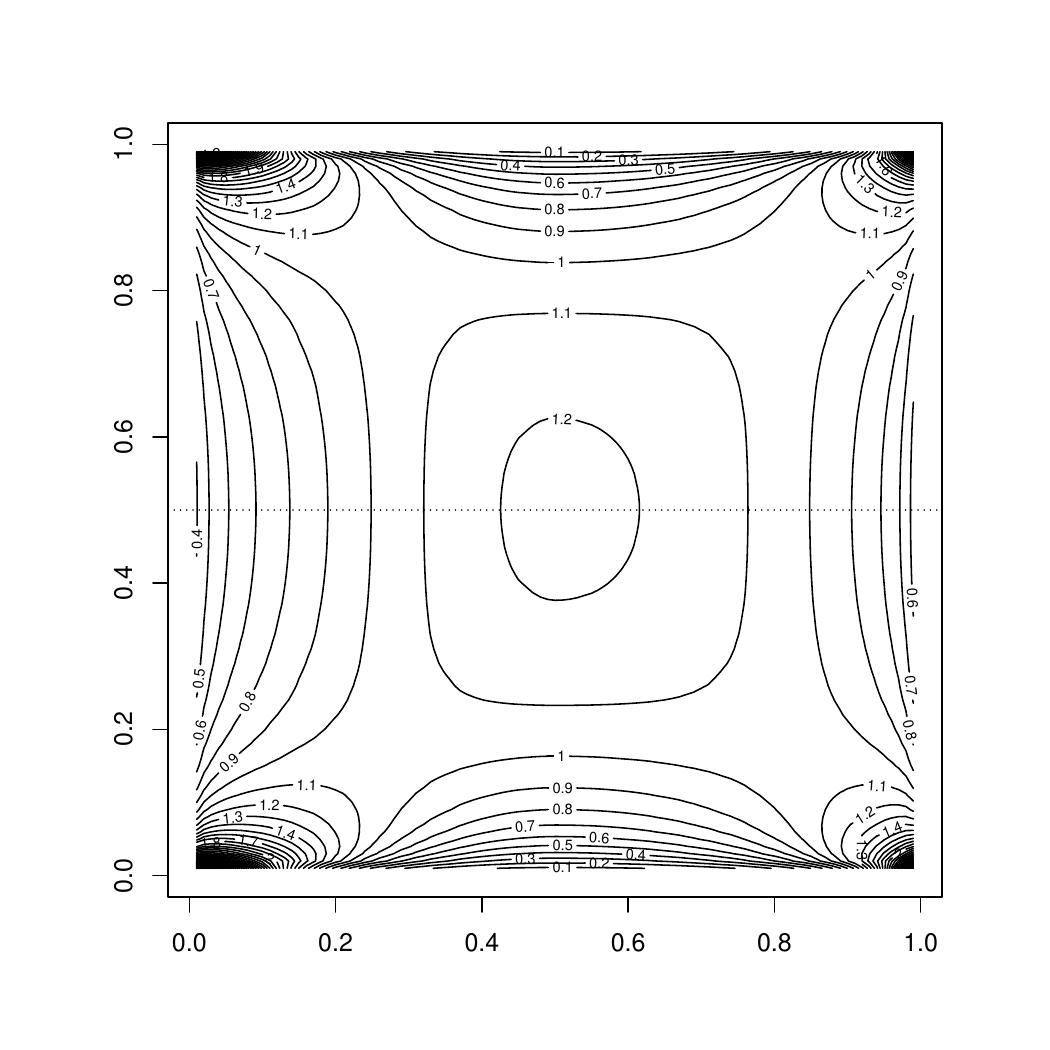}
\includegraphics[width=7cm,height=7cm]{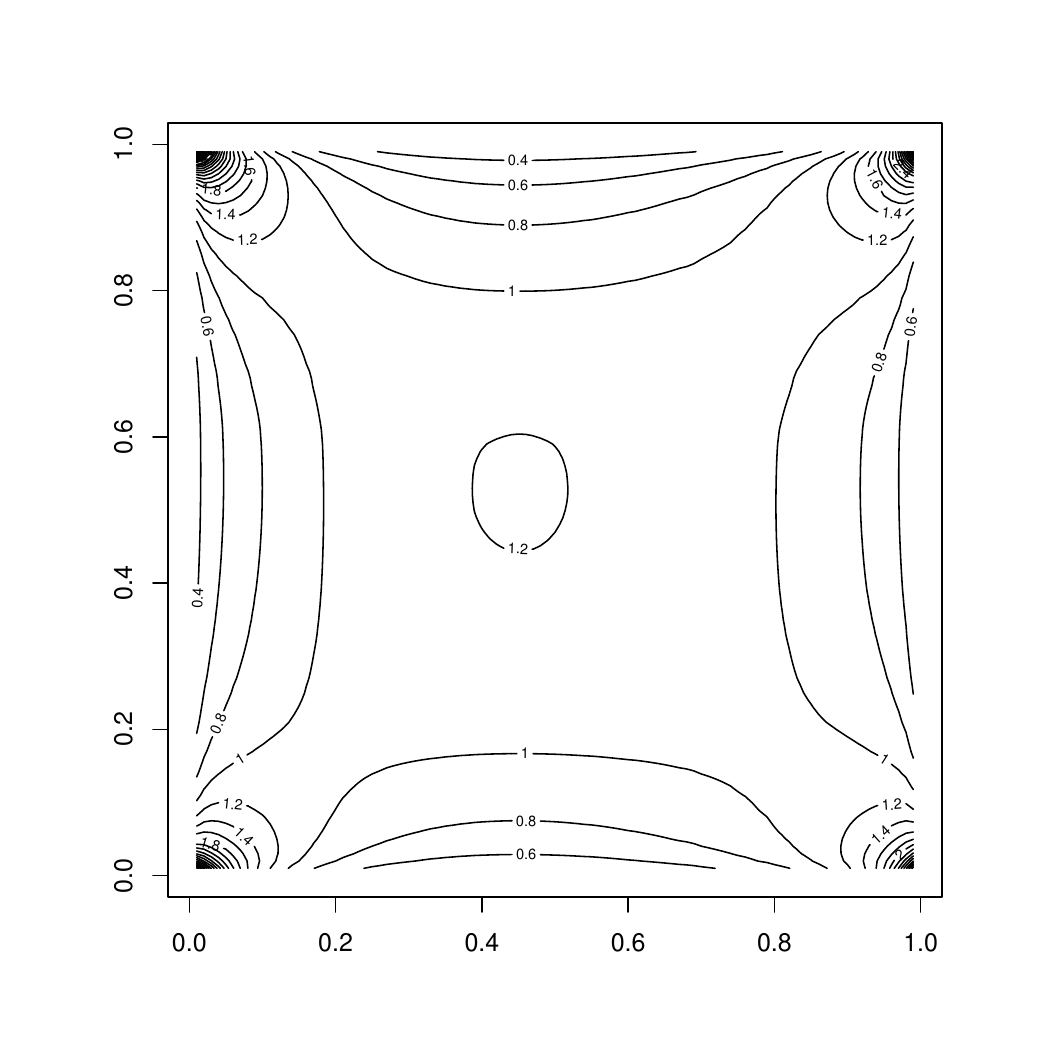}
\caption{Introducing asymmetry.
Left: ARCH(1) with Gaussian innovations in which the volatility function has leverage and the
parameters are $\alpha_0 =0.4$, $\alpha_1 = 0.3$ and $\gamma_1 = 0.4$.
Right: classical ARCH(1) model with 4 degrees of freedom
and skewness parameter $\lambda = 0.8$ where $\alpha_0 = 0.4$ and
$\alpha_1 = 0.6$.
}
\label{fig:2}
\end{figure}
  
\subsection{Higher-order copulas of GARCH(1,1)-type
  processes}\label{sec:high-order-cond}


\newtext{In this section we show that the same symmetries evident in
  the copula density $c_1$ can be found in the density $c_k$ of the copula of the conditional distribution of
$(X_{1},X_{k+1})^\top$ given 
$\bm{X}_{[2:k]} = \bm{x}_{[2:k]}$ for $k\geq2$, 
which takes
the form}
\begin{equation}\label{eq:2}
  c_k(u,v \mid \bm{x}_{[2:k]}) = \frac{f_{X_1,X_{k+1} \mid \bm{X}_{[2:k]}}\left(
    F_{X_1 \mid \bm{X}_{[2:k]}}^{-1}\left(u \mid
      \bm{x}_{[2:k]}\right),  F_{X_{k+1} \mid \bm{X}_{[2:k]}}^{-1}\left(v \mid
      \bm{x}_{[2:k]}\right) \mid \bm{x}_{[2:k]} \right)}
{ f_{X_1 \mid \bm{X}_{[2:k]}} \left(F_{X_1 \mid \bm{X}_{[2:k]}}^{-1}\left(u \mid
      \bm{x}_{[2:k]}\right)\mid  \bm{x}_{[2:k]} \right)
  f_{X_{k+1} \mid \bm{X}_{[2:k]}} \left(F_{X_{k+1} \mid \bm{X}_{[2:k]}}^{-1}\left(v \mid
      \bm{x}_{[2:k]}\right)\mid  \bm{x}_{[2:k]} \right)
  }
\end{equation}
\newtext{It should be distinguished from the copula density $c_{(k)}$
  of the copula $C_{(k)}$ of the marginal distribution of $(X_1,
  X_{k+1})^\top$ (unless $k=1$). In particular, it is the densities
  $c_k$ rather than $c_{(k)}$ that appear in the D-vine
  representation~\eqref{eq:27} of the GARCH-type process. The former depend on the conditioning
  variables $\bm{x}_{[2:k]}$ although we would like to be able to assume
  that this dependence is weak so that the D-vine representation can be approximated by a simplified D-vine.}

 \newtext{To understand the main argument in this section it suffices
   to consider the case $k=2$. We consider a process
with an explicit GARCH effect since, for a process of ARCH(1) type,  $X_1$ and $X_3$ are
conditionally independent given $X_2=z$ and 
 $c_2(u,v \mid z) \equiv 1$, the density of the independence copula. The joint
conditional density $f_{X_1,X_{3}\mid X_2}(x,y \mid z)$ and the joint
density $f_{X_1,X_3}(x,y)$ can both be calculated from the
joint density
$ f_{X_1,X_2,X_3}(x,z,y)$ by observing that $f_{X_1,X_{3}\mid X_2}(x,y\mid z) =
f_{X_1,X_2,X_3}(x,z,y)/f_X(z)$ and $f_{X_1,X_{3}}(x,y) =
\int_{-\infty}^\infty f_{X_1,X_2,X_3}(x,z,y)\, \rd z$.
Since
\begin{align*}
  f_{X_1,X_2,X_3}(x,z,y)
                         &= \int_0^\infty
                           f_{X_3 \mid X_2, X_1,\sigma_1}(y \mid z,x,s)
f_{X_2 \mid X_1,\sigma_{1}}(z \mid x, s) f_{X_1 \mid \sigma_1}(x
\mid s ) f_\sigma(s)
\rd s \\
                         &= \int_0^\infty
\frac{1}{\volfunc(z,\volfunc(x,s))}f_\epsilon\left(  \frac{y}{\volfunc(z,\volfunc(x,s))}    \right)
                           \frac{1}{\volfunc(x,s)} f_\epsilon\left(
     \frac{z}{\volfunc(x,s)} \right)
 \frac{1}{s}
      f_{\epsilon}\left(\frac{x}{s}\right) f_\sigma(s) \rd s
\end{align*}
we can infer that $f_{X_1,X_2,X_3}(x,z,-y) = f_{X_1,X_3,X_2}(x,z,y)$
for a process with symmetric innovations and
$f_{X_1,X_2,X_3}(-x,z,y) =
f_{X_1,X_2,X_3}(x,z,y)$ for a symmetric process. It follows
immediately that the conditional
and unconditional densities satisfy $f_{X_1,X_{3}
    \mid X_2}(x, - y \mid z) = f_{X_1,X_{3}
    \mid X_2}(x, y \mid z)$ and  $f_{X_1,X_{3}
   }(x, - y) = f_{X_1,X_{3}
  }(x, y)$ for a process with symmetric innovations and we have h-symmetry in both cases. Moreover, they
  satisfy $f_{X_1,X_{3}
    \mid X_2}(-x, y \mid z) = f_{X_1,X_{3}
    \mid X_2}(x, y \mid z)$ and  $f_{X_1,X_{3}
   }(-x, y) = f_{X_1,X_{3}
  }(x, y)$ for a symmetric process and we have joint symmetry in both
  cases.
  By Lemma~\ref{lemma:symm-copulas}, the various forms of 
symmetry are inherited by the conditional
copula $C_2$ of $(X_1,X_3)^\top$ given $X_2$ and the marginal copula 
$C_{(2)}$ of $(X_1,X_3)^\top$.
}

\newtext{We now give the general result for $k \geq 2$. Since the
  logic of its proof is identical to the argument given above, albeit
  with more cumbersome notation, the general proof is given in the
  Supplementary Material.}
\begin{proposition}\label{prop:symmetry2}
\begin{enumerate}[label={(\roman*)},itemindent=1em]
\item In a GARCH(1,1)-type process with symmetric distribution, the
  conditional copula $C_k$ and the marginal copula $C_{(k)}$ are h-symmetric about zero
  for $k\geq2$.\label{thm:2symmetry1}
\item In a symmetric GARCH(1,1)-type process $C_k$ and $C_{(k)}$ are
  jointly symmetric about the origin for $k\geq2$.\label{thm:2symmetry2}
\end{enumerate}
\end{proposition}

\newtext{The computational tools we have developed to approximate
  $c_1$ in
Section~\ref{sec:comp-exampl}
can be extended to approximate $c_2$, and we use them to explore the simplifying
assumption. Let $c_2(u,v ;
w) := c_2(u,v \mid F_X^{-1}(w))$ for $w \in (0,1)$.}
Contour plots of $c_1(u,v;w)$ for the cases where
$w=0.5$ and $w=0.75$ are shown in Figure~\ref{fig:3}. The two pictures are
very similar to each other in terms of the shapes and heights of
the contour lines, and are also similar in pattern to the plot of $c_1$ in
the lower-right panel of Figure~\ref{fig:1}.
We compute an approximate value for
the absolute distance between the conditional copula densities and the
independence copula density using the measure $D_2(w) = \int_0^1\int_0^1 | c_2(u,v ; w) -1| \rd u \rd v$.
We obtain the estimates $D_2(0.5) = 0.10$ and
$D_2(0.75) = 0.09$, which can be compared with the equivalent measure
$D_1 = 0.16$ for the copula density $c_1$ in the lower-right panel of Figure~\ref{fig:1}.

\begin{figure}[ht!]
\centering
\includegraphics[width=6.5cm,height=6.5cm]{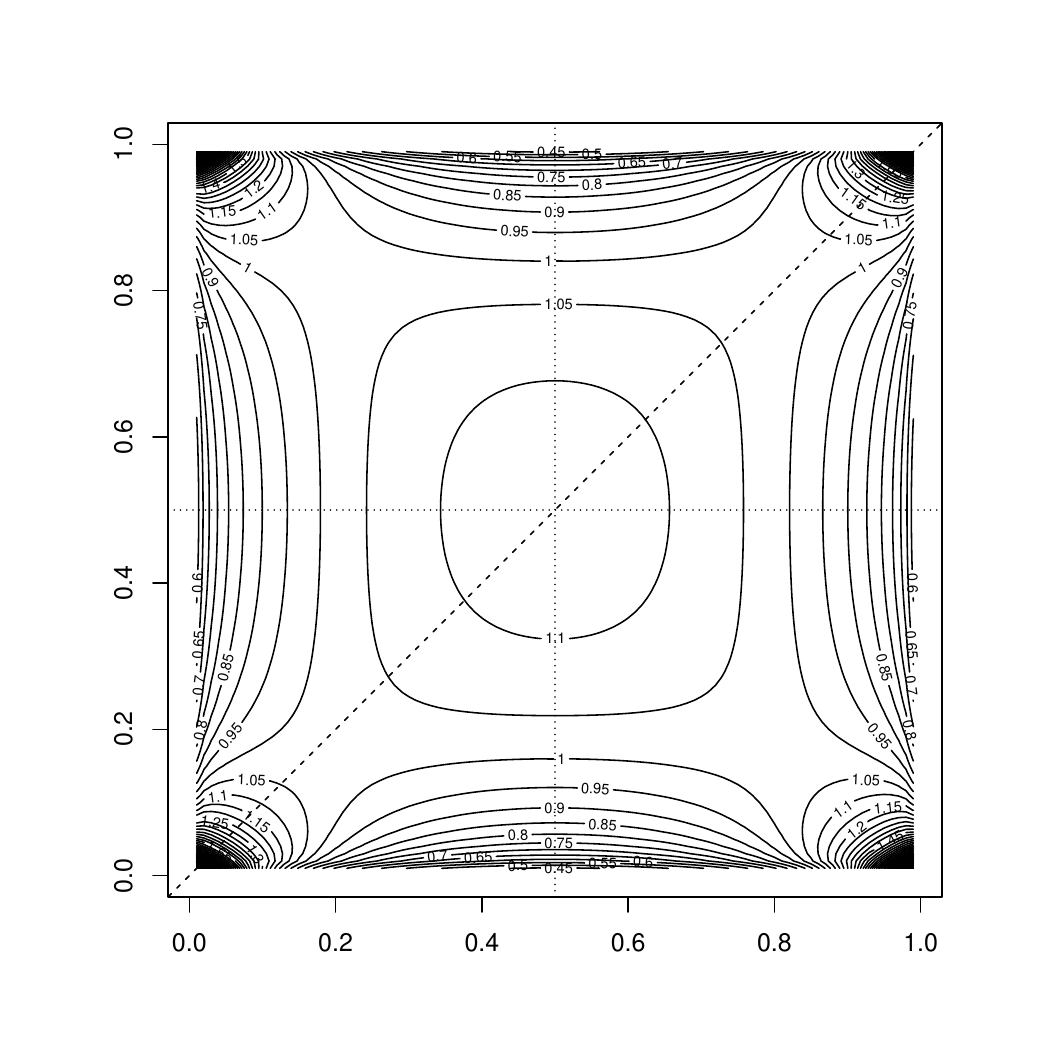}
\includegraphics[width=6.5cm,height=6.5cm]{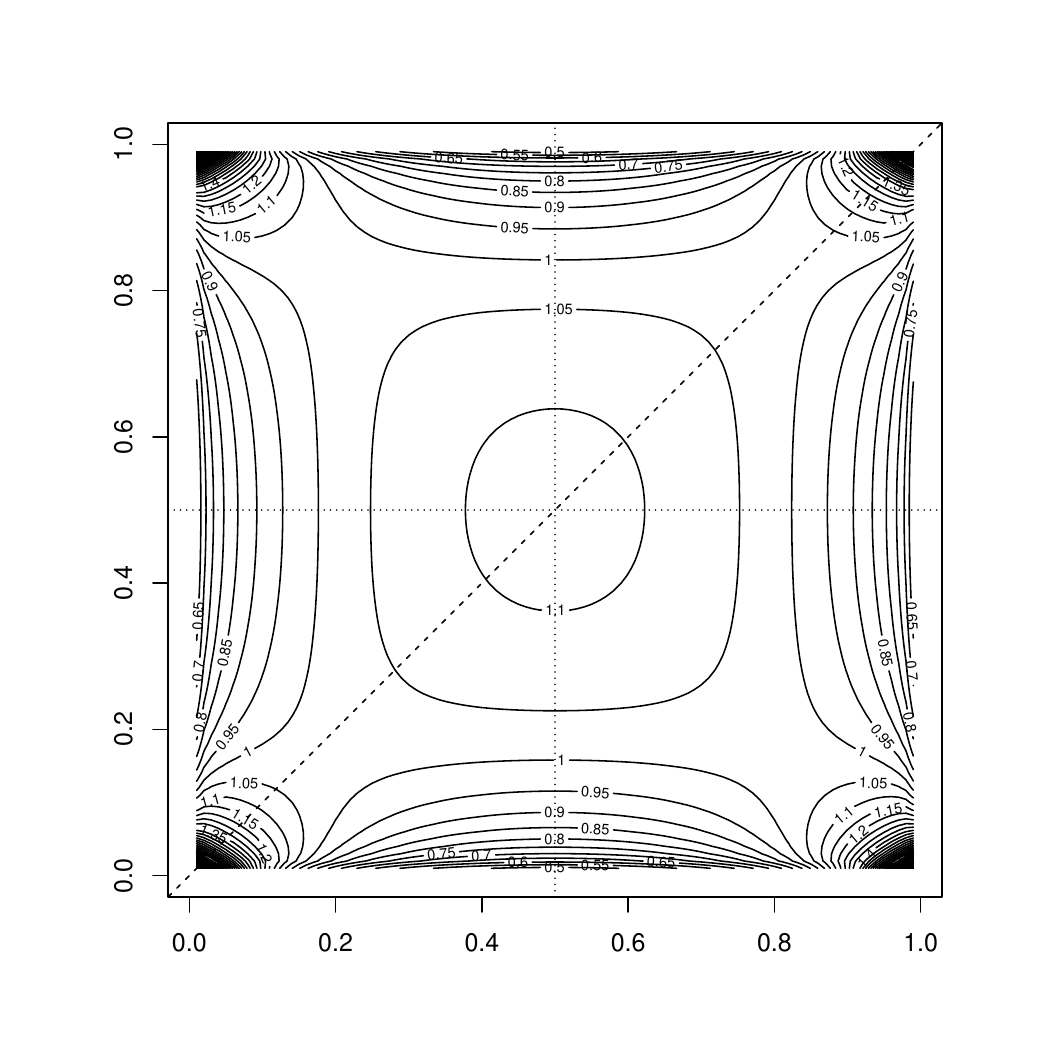}
 \caption{Contour plots of $c_2(u,v; w)$ for the classical
   GARCH(1,1) process with Gaussian innovations. Left: $w=0.5$. Right:
   $w=0.75$.
}
\label{fig:3}
\end{figure}

As might be expected, the dependence weakens from lag 1 to lag
2. Moreover, the value of $w$ does not seem to
have a particularly large effect on the conditional copula at lag
2. This might be taken as weak evidence in favour of the simplifying
assumption. \newtext{It would be desirable to look at
values of $w$ further into the tail and consider lags $k >2$, but our approximation methods
then become less accurate leading to less confident conclusions.}



 \section{Using v-transforms to approximate GARCH pair copulas}\label{sec:copula-dens-constr}
\subsection{V-transforms and their inverses}

\newtext{It has long been observed that serial dependence in financial
asset-price return series, and in the models like GARCH that describe them,
shows up as serial correlation of absolute or squared
values. \cite{bib:loaiza-maya-et-al-18} considered the componentwise application of
symmetric v-shaped
transformations $T$ (such as the absolute value function $T(x) = |x|$)
to a pair of random variables $(Y,Z)^\top$ and derived an expression for the relationship of the
copula $C_T$ of $(T(Y),T(Z))^\top$ to the copula
$C$ of $(Y,Z)^\top$. They showed that the resulting copula $C_T$ involves the
marginal distributions $F_Y$ and $F_Z$ except in the special case of
symmetric margins. \cite{bib:mcneil-20} observed that the effect on the copula of
such transformations, as well as more general asymmetric
transformations, can be isolated by introducing v-shaped transformations
of uniform random variables that are
  uniform-distribution-preserving~\citep{bib:porubsky-salat-strauch-88}. Importantly, for the
  purposes of this paper, these so-called
  v-transforms can be inverted stochastically in a way that also
  preserves uniform distributions. We use this inversion technique to create new
  copulas that show behaviour similar to GARCH pair copulas.
In this section we recap the
  essential elements of the theory.}

\begin{definition}\label{theorem:v-characterization}
A mapping $\vtrans: [0,1] \to [0,1]$ is a v-transform if it takes the form
  \begin{equation}
    \label{eq:1_v}
    \vtrans(u) =
    \begin{cases}
(1-u) - (1-\downprob) \Psi \left( \frac{u}{\downprob} \right) & u \leq
\downprob, \\
u - \downprob \Psi^{-1}\left( \frac{1-u}{1-\downprob} \right) & u > \downprob,
\end{cases}
\end{equation}
for a parameter $0 < \delta 
< 1$, known as the fulcrum, and a continuous and strictly increasing distribution
function $\Psi$ on $[0,1]$, known as the generator.
\end{definition}

Mappings of this form are always v-shaped with $\vtrans(0) =
\vtrans(1) = 1$ and $\vtrans(\delta) = 0$. They preserve uniformity in the sense that if $U
\sim \mathcal{U}(0,1)$ is standard uniform, then $\vtrans(U) \sim \mathcal{U}(0,1)$. If the generator
is $\Psi(x) = x$ (the uniform df), we
obtain the linear v-transform, which satisfies $\vtrans(u; \delta) = (\delta-u)/\delta$ for $u \leq \delta$ and
$\vtrans(u; \delta) = (u-\delta)/(1-\delta)$ for $u > \delta$; this
includes the special case
$\vtrans(u; 0.5) = |2u-1|$, which is the symmetric v-transform. It is possible
to introduce further parameters. For example, by taking $\Psi(x) =
x^\kappa$ for $\kappa >0 $, we obtain a two-parameter family
  \begin{equation}
    \label{eq:2_v}
    \vtrans(u; \delta,\kappa) =
    \begin{cases}
(1-u) - (1-\downprob) \left(\frac{u}{\delta} \right)^\kappa& u \leq
\downprob, \\
u - \downprob \left( \frac{1-u}{1-\delta}\right)^{1/\kappa} & u > \downprob.
\end{cases}
\end{equation}

For further results we need some more notation.
For every $y \in (0,1]$ there are two possible
values $u$ on either side of the fulcrum $\delta$ such that $\vtrans(u) = y$ and
these may be written as $\vtrans^{-1}(y)$ and $\vtrans^{-1}(y) +y$ where
$\vtrans^{-1}(y) = \inf\{u : \vtrans(u) =y \}$ is a partial inverse of
$\vtrans$. We also introduce the function
$\Delta(y) = 
- 1/\vtrans^\prime(\vtrans^{-1}(y))$ where
$\vtrans^\prime$ is the derivative of $\vtrans$. If $U \sim \mathcal{U}(0,1)$
then~\cite{bib:mcneil-20} shows that
$\P(U \leq \delta \mid \vtrans(U) = y) = \Delta(y)$ for $y \in (0,1]$; for
the linear v-transform $\Delta(y) =
\delta$, but
otherwise $\Delta(y)$ depends explicitly on $y$.

\newtext{
We can stochastically invert a v-transform in such a way that we also
preserve uniformity of random variables. For $y \in (0,1]$, the stochastic inverse
$\vtrans^\leftarrow(y, W)$ of the v-transform $\vtrans$ is defined by
taking a uniform random variable $W\sim \mathcal{U}(0,1)$ and setting
$\vtrans^\leftarrow(y, W) =\vtrans^{-1}(y)$ 
if $W \leq \Delta(y)$ and $\vtrans^\leftarrow(y, W) =\vtrans^{-1}(y) +
y$ otherwise; we always set $\vtrans^\leftarrow(0,W) = \vtrans^{-1}(0)
= \delta$. If $U$ and $W$ are iid
$\mathcal{U}(0,1)$ variables, then it can also
be shown that $\vtrans^\leftarrow (U,W) \sim \mathcal{U}(0,1)$. Note, however, that while $\vtrans\left(\vtrans^\leftarrow (y,W)\right) = y$
 holds for all $y\in [0,1]$, the equation $\vtrans^\leftarrow\left(
 \vtrans(u),W\right) = u$ is only true with probability $\Delta(u)$, if
 $u< \delta$, and with probability $1-\Delta(u)$, if $u > \delta$.}


 \subsection{V-transformation and inverse v-transformation of copulas}

\newtext{The following result, which is an adaptation of a result in~\cite{bib:mcneil-20}, is key to constructing new copulas using
v-transforms and their inverses. Part (i) is related to the
  result of~\cite{bib:loaiza-maya-et-al-18} but is expressed in a
  manner that involves only copulas and is free of marginal
  distributions.}
\begin{theorem}\label{theorem:v-and-inverse-v}
Let $\mathcal{V}_1$ and $\mathcal{V}_2$ be v-transforms with
associated partial inverses $\mathcal{V}_1^{-1}$ and
$\mathcal{V}_2^{-1}$, $\Delta$-functions $\Delta_1$ and
$\Delta_2$ and stochastic inverses  $\vtrans_1^\leftarrow$ and $\vtrans_2^\leftarrow$. Let $(U,V)^\top$ be
distributed according to a copula $C$ with density $c$
\newtext{and let $W_1$ and $W_2$ be iid $\mathcal{U}(0,1)$ variables,
  independent of $(U,V)^\top$.}
\begin{enumerate}[label={(\roman*)},itemindent=1em]
\item
$(\vtrans_1(U),\vtrans_2(V))^\top$ is distributed according to a copula
$C^*$ with density
\begin{align}
& c^*(u,v) = 
 c(\mathcal{V}_1^{-1}(u),
\mathcal{V}_2^{-1}(v))\; \Delta_1(u) \,\Delta_2(v) \; + c(\mathcal{V}_1^{-1}(u)
+ u,
\mathcal{V}_2^{-1}(v)) \; \big(1-\Delta_1(u)\big)\,\Delta_2(v) \; + \nonumber \\
& c(\mathcal{V}_1^{-1}(u),
\mathcal{V}_2^{-1}(v)+v) \; \Delta_1(u)\,\big(1-\Delta_2(v)\big)  +
 c(\mathcal{V}_1^{-1}(u)+u,
\mathcal{V}_2^{-1}(v) +v) \; \big(1-\Delta_1(u)\big)\,
  \big(1-\Delta_2(v)\big). \label{eq:v-transform-density}
\end{align}
\item
  $(\vtrans_1^\leftarrow(U, W_1),\vtrans_2^\leftarrow(V, W_2))^\top$ is distributed according to a copula
$C^\dagger$ with density 
\begin{equation}
  \label{eq:inverse-v-transform-density}
 c^\dagger(u,v)
= c(\mathcal{V}_1(u), \mathcal{V}_2(v)).
\end{equation}
\end{enumerate}
\end{theorem}
\begin{proof}
  This may be proved by a simple adaptation of the proof of
  Theorem 3 in~\cite{bib:mcneil-20} which deals with the case of two
  (or more)
  identical v-transforms.
\end{proof}

\noindent \newtext{
Part (i) refers to what we call
  (componentwise) v-transformation of $(U,V)^\top$, or its distribution 
$C$, while 
part (ii) refers to what we call (componentwise independent) 
inverse v-transformation of $(U,V)^\top$, or $C$. Inverse
v-transformation of $(U,V)^\top$ followed by v-transformation returns the original
random vector $(U,V)^\top$ (since 
$\vtrans(\vtrans^{\leftarrow}(y,W)) = y$ for $y \in [0,1]$ and any v-transform). However
v-transformation of $(U,V)^\top$ 
followed by inverse
v-transformation returns a different random vector with, in general,
a different distribution. An exception is the following special case.}

\begin{corollary}\label{cor:newnew}
  \newtext{
Under the conditions of Theorem~\ref{theorem:v-and-inverse-v}, if $C$ is jointly symmetric and $\vtrans_1(u) =
\vtrans_2(u) = |2u-1|$ for $u \in [0,1]$, then
$(\vtrans_1^\leftarrow(\vtrans_1(U),
W_1),\vtrans_2^\leftarrow(\vtrans_2(V), W_2))^\top \eqdis (U,V)^\top$.
}
\end{corollary}
\begin{proof}
  \newtext{
  By part (ii) of Corollary~\ref{prop:copula-absolute-cor} the density $c$ of the copula $C$ of $(U,V)^\top$
  can be written as
  $c(u,v) = c^*(\vtrans_1(u),\vtrans_2(v))$ where $c^*$ is the density
  of the copula of $(\vtrans_1(U),\vtrans_2(V))^\top$. By part (ii) of
  Theorem~\ref{theorem:v-and-inverse-v} $c^*(\vtrans_1(u),\vtrans_2(v))$ is also the density of
  the inverse v-transformation of
  $(\vtrans_1(U),\vtrans_2(V))^\top$ with the stochastic inverses
  $\vtrans_1^\leftarrow$ and $\vtrans_2^\leftarrow$.}
  \end{proof}

\subsection{Implicit v-transforms in general bivariate
  densities}\label{sec:mimick-non-refl}

In this section we use part (i) of
Theorem~\ref{theorem:v-and-inverse-v} to find a condition
that allows us to write
the copula density $c$
of a random vector $(Y,Z)^\top$ on $\R^2$ in the form $c(u,v) =
c^*(\vtrans_1(u),\vtrans_2(v))$ for v-transforms $\vtrans_1$ and
$\vtrans_2$ and a copula density $c^*$. \newtext{This result requires the
concept of a v-shaped transformation on $\R$, which generalizes the
idea found in~\cite{bib:loaiza-maya-et-al-18} to also allow non-symmetric transformations.}

\begin{definition}
A function $T$ is a v-shaped function on $\R$ if it can be written in the
form $ T(x) =
T_L(\mu_T-x)$ for $x \leq \mu_T$ and $T(X) = T_R(x-\mu_T)$ for $x > \mu_T$
where $T_L$ and $T_R$ are strictly
increasing, continuous and differentiable functions on
$\R^+=[0,\infty)$ such that $T_L(0) = T_R(0)$ and $\mu_T \in
\R$. 
\end{definition}

\begin{theorem}\label{prop:new-rep}
Let $(Y,Z)^\top$ be a random vector in $\R^2$ with joint density
$f_{Y,Z}$, marginals densities $f_Y$ and $f_Z$ and copula density $c$. If the density ratio
$f_{Y,Z}(y,z)/(f_Y(y) f_Z(z))$
can be written in the form
\begin{equation}\label{eq:5}
 \frac{f_{Y,Z}(y,z)}{f_{Y}(y)f_Z(z)} = g\left(T_1(y), T_2(z)\right)
\end{equation}
for two v-shaped transformations $T_1$ and $T_2$ on $\R$ and some
function $g: \R \times \R \to \R^+$, then we may write
$c(u,v) = c^*(\vtrans_1(u), \vtrans_2(v))$ on $[0,1]^2$ where
$\vtrans_1$ and $\vtrans_2$ are v-transforms given by
$\vtrans_1(u) = F_{T_1(Y)} \circ T_1 \circ F_Y^{-1}(u)$ and
$\vtrans_2(v) = F_{T_2(Z)} \circ T_2 \circ F_Z^{-1}(v)$ and where
$c^*$ is the density of the copula of $(T_1(Y), T_2(Z))^\top$.
\end{theorem}
\begin{proof}
The density $c$ is the density of the df of the random variables
$(U,V)^\top$ where $U = F_Y(Y)$ and $V = F_Z(Z)$. In~\cite{bib:mcneil-20} it is shown that
$F_{T_1(Y)}(T_1(Y)) = \vtrans_1(U)$ for the v-transform given by
$\vtrans_1 = F_{T_1(Y)} \circ T_1 \circ F_Y^{-1}$; similarly $F_{T_2(Z)}(T_2(Z)) = \vtrans_2(V)$ for 
$\vtrans_2 = F_{T_2(Z)} \circ T_2 \circ F_Z^{-1}$. Hence we may write
\begin{equation}\label{eq:26}
  c(u,v) = g \left( T_1(F_Y^{-1}(u)), T_2(F_Z^{-1}(v))  \right) =g\left(F_{T_1(Y)}^{-1}(\vtrans_1(u)),F_{T_2(Z)}^{-1}(\vtrans_2(v))\right).
\end{equation}

Let $c^*$ be the density of the df of
$(\vtrans_1(U),\vtrans_2(V))^\top$ and thus the density of the copula
of $(T_1(Y),T_2(Z))^\prime$.
We can compute $c^*$ from~\eqref{eq:26}
using~(\ref{eq:v-transform-density}). Since
$\vtrans_1(\vtrans_1^{-1}(u)) = \vtrans_1(\vtrans_1^{-1}(u) +u) =u$ and
$\vtrans_2(\vtrans_2^{-1}(v)) = \vtrans_2(\vtrans_2^{-1}(v) +v) = v$ we infer that
$c^*(u,v) = g\left(F_{T_1(Y)}^{-1}(u),F_{T_2(Z)}^{-1}(v)\right)$
and
hence it follows that $c(u,v) =
 c^*(\vtrans_1(u),\vtrans_2(v))$.
\end{proof}

Symmetric models of GARCH(1,1)-type provide examples of models where
the density ratio for $f_{X_1,X_2}$ (and also for $f_{X_1,X_{k+1} \mid
  \bm{X}_{[2:k]}}$) may be factorized in the
form~\eqref{eq:5} by taking $T_1(x) = T_2(X) = |x|$. Unfortunately, when innovations distributions are
asymmetric, or when leverage is introduced into a model 
that is not of ARCH(1) type effect, perfect factorization of the
form~\eqref{eq:5} does not seem to be possible. 
However, it is possible in the following case.


\begin{example}[ARCH(1) process with symmetric innovations and leverage]\label{ex:arch-leverage}
For the process
of ARCH(1) type with leverage and symmetric distribution in the upper
left of Figure~\ref{fig:2} we observe that the
volatility function $\volfunc(x)$ is itself a v-shaped
transformation with $\mu_\sigma=0$, $\sigma_L(x) = \sqrt{\alpha_0 + (\alpha_1 + \gamma_1)x^2}$
and $\sigma_R(x) = \sqrt{\alpha_0 + \alpha_1 x^2}$.
In this case, it follows from the joint density in~(\ref{eq:22}) that we may write
\begin{displaymath}
  \frac{f_{X_1,X_2}(x,y)}{f_X(x)f_X(y)} = \frac{f_\epsilon\left(
      \frac{|y|}{\sigma(x)} \right)}{\sigma(x) f_X(|y|)}= g(\sigma(x), |y|)
\end{displaymath}
and conclude that the copula density satisfies $c_1(u,v) =
c_1^*(\vtrans_1(u),|2v-1|)$ where $c_1^*$ is the copula density of
$(\volfunc(X_1), |X_2|)^\top$ or, in other words,
$(\sigma_t,|X_t|)^\top$ for any $t$. To derive $\vtrans_1$ we  must
find the composition of functions $\vtrans_1(u) = F_{\sigma} \circ
\volfunc \circ F_X^{-1}(u)$.
By using the fact that the distribution function of $\sigma_t$ is
$ F_\sigma(s) = F_X\left(\sqrt{(s^2-\alpha_0)/\alpha_1}\right) -F_X\left(-\sqrt{(s^2-\alpha_0)(\alpha_1+\gamma_1)}\right)$
we obtain after some calculation that
 \begin{displaymath}
    \vtrans_1(u) =
\begin{cases}
F_X\left(-\left(1 + \frac{\gamma_1}{\alpha_1}\right)^{\frac{1}{2}} F_X^{-1}(u) \right) -
u, &u \leq 0.5, \\
 u - F_X\left(-\left(1 + \frac{\gamma_1}{\alpha_1}\right)^{-\frac{1}{2}} F_X^{-1}(u)\right) ,& u > 0.5,
\end{cases}
\end{displaymath}
which is an asymmetric function for $\gamma_1 \neq 0$.
\end{example}

\subsection{Approaches to approximating GARCH copulas}\label{sec:appr-mimick-garch}

\newtext{Despite the limited number of GARCH(1,1)-type models for which we can
prove that pair copulas must take the exact form $c(u,v) =
c^*(\vtrans_1(u),\vtrans_2(v))$, this construction remains of interest as
a flexible means of approximating the cross-like behaviour of the
implied GARCH pair copulas. Our approach to this problem is to use simulation
studies to first identify copula densities of the form $c^*(|2u-1|,
|2v-1|)$ to approximate the jointly symmetric copulas of symmetric GARCH
processes and then to
consider how asymmetric
v-transforms $\vtrans_1$ and $\vtrans_2$ can accommodate the
asymmetries introduced by leverage and asymmetric innovations.}


\newtext{
There are two ways of identifying
appropriate copula densities $c^*$ for building models with densities of the
form $c^*(|2u-1|,|2v-1|)$: one is the explicit method using part (ii)
of Theorem~\ref{theorem:v-and-inverse-v} where
we simply choose a candidate for the copula $C^*$ based on the insights we have collected
so far; the other is the implicit method where we consider constructions that are known
to yield a jointly symmetric copula and then infer the form of $c^*$.}

Explicit choices of $C^*$ might be the Gumbel,
Joe or
180-degree-rotated (survival) Clayton copulas. These all have upper
tail dependence and may be capable of approximating the copulas of the
absolute value processes of symmetric GARCH(1,1)-type models. Recall
that the densities of such copulas can be thought of as blown-up
versions of the upper-right quadrants of the copula densities in
Figure~\ref{fig:1}. \newtext{We impose one technical condition on
  copulas $C^*$ which 
  facilitates maximum likelihood inference for models with densities
  of the form $c^*(\vtrans_1(u),\vtrans_2(v))$: we
  require that $c^*(0,0) <
  \infty$ which ensures that the log-likelihood is properly defined for observations
  at the fulcrum values.}

As examples of the
implicit construction we could consider:
\begin{enumerate}[label={(\arabic*)},itemindent=1em]
\item the copula of a spherically symmetric distribution;
\item an equal (50/50) mixture of a radially symmetric copula and its
90-degree rotation;
\item an equal (25/25/25/25) mixture of any copula and its rotations through 90, 180 and
270 degrees.
\end{enumerate}

An example of (1) is the copula of a spherical t distribution, i.e.~a
bivariate t distribution located at the origin with identical
marginal distributions and correlation equal to zero. We call this
copula the spherical t copula and denote it
$C^t_\nu$.  Since a spherical t-distributed random vector $(Y,Z)^\top$
is jointly symmetric, so is its copula. As the
degree-of-freedom parameter tends to zero, this copula becomes
increasingly cross-shaped and converges to an equal mixture of the
comonotonicity and countermonotonicity
copulas~\citep{bib:mcneil-neslehova-smith-22}.
We denote the copula of $(|Y|,|Z|)^\top$ by $C^{|t|}_\nu$ and call it
the absolute spherical t copula (ast). By~\eqref{eq:8} its density must satisfy
\begin{displaymath}
    c^{|t|}_\nu(u,v) = c^t_\nu\left(\frac{1-u}{2},\frac{1-v}{2}\right) =
  c^t_\nu\left(\frac{1-u}{2},\frac{1+ v}{2}\right) =
  c^t_\nu\left(\frac{1+u}{2},\frac{1-v}{2}\right) = c^t_\nu\left(\frac{1+u}{2},\frac{1+v}{2}\right).
\end{displaymath}
This one-parameter copula, which appears to be new, interpolates between independence (as $\nu \to \infty$) and
comonotonicity (as $\nu \to 0$). It has upper tail dependence like the
Gumbel, Joe and survival Clayton
copulas. More information is collected in Section~\ref{sec:ast-copula} in the
Supplementary Material.


\newtext{An example of (2) is an equal mixture of the t copulas $C^t_{\nu,\rho}$
and $C^t_{\nu,-\rho}$ for some $\rho \geq 0$, which is the jointly symmetric
special case of a copula considered
by~\citet{bib:loaiza-maya-et-al-18}; the t mixture reduces to the
spherical t copula $C^t_\nu$ when $\rho=0$. Construction (3), which is the bivariate version
of a
higher-dimensional family used
in~\citet{bib:oh-patton-16}, could be based on the
Gumbel, Clayton or Joe copulas, to take advantage of the strong tail
dependence of these copulas in either the upper-right or lower-left
corners. The Gumbel case of (3) is the jointly symmetric special case of the
convex Gumbel copula proposed by~\citet{bib:junker-may-05}
and used
by~\citet{bib:loaiza-maya-et-al-18} in their analysis}.

\newtext{
So far none of these constructions address the issue of
non-exchangeability which was clearly evident in certain examples in
Figure~\ref{fig:1}.
A possible construction for non-exchangeable copulas was proposed
by~\cite{bib:khoudraji-95}
 and we use a version of this construction in which we generalize an exchangeable copula
$C_\theta$ with parameter $\theta$ 
by setting
\begin{equation}
  \label{eq:6}
  C_{\theta,a_1,a_2}(u,v) = C_\theta(u^{1-a_1}, v^{1-a_2}) u^{a_1}v^{a_2}
\end{equation}
for $a_1 \in [0,1]$ and $a_2\in [0,1]$.
The resulting
copula is
non-exchangeable when $a_1 \neq a_2$. }
We apply the Khoudraji extension to copulas before either mixing them
with their rotations or applying inverse v-transformation.


\subsection{Simulation study of the approximation of $C_1$}\label{sec:results-siul-study}

\newtext{To get some insight into the
  approximations that work best for the pair copulas implied by
  GARCH-type models, we carry out a simulation study in which we focus
  on estimation of the first-order copula $C_1$. Samples of
50000 values are
generated from ARCH(1)-type
and GARCH(1,1)-type processes using the parameter values and innovation
distributions summarized in Table~\ref{tab:parameters}.
The copula $C_1$ is estimated by fitting a first-order D-vine
to the data using the two-stage method developed by~\cite{bib:chen-fan-06} in
which the marginal distribution of the process is estimated  by the
rescaled empirical distribution function (EDF) to produce pseudo-copula data. The explicit and
implicit constructions of Section~\ref{sec:appr-mimick-garch} are
tested and fitted models are
compared according to their AIC values. Note that, the first-order
  copulas $C_1$ are the only pair copulas for which we can produce
  pseudo-copula data using only the rescaled EDF; we address 
  approximations for higher-order copulas $(C_k)_{k>1}$ indirectly in Section~\ref{sec:examples} by fitting
  higher-order simplified D-vines to simulated GARCH data.}

\newtext{Under the assumption that the
model is first-order Markov with known copula,~\citet{bib:chen-fan-06} show consistency
and asymptotic normality of the resulting parameter estimates subject to technical conditions which include that the process should be $\beta$-mixing at
a geometric rate. All the GARCH(1,1)-type
processes we consider in our simulation satisfy this
mixing condition~\citep{bib:francq-zakoian-06}. However, the first-order Markov assumption is only true for ARCH(1)-type processes
and our models contain approximating copulas rather than correctly
specified copulas. \cite{bib:nagler-kruger-min-20} supply some very
general results which can be applied to
semiparametric
inference for simplified D-vines of higher order when the
pair copulas are misspecified. Their results suggest
that consistency and asymptotic normality of the procedure 
hold in the misspecified model, albeit convergence is to the parameter
values that solve the model score equations, which coincide with the
true values in a correctly specified model.
}

\subsubsection{Jointly symmetric models}\label{sec:jointly-symm-case}

                               \begin{table}[!ht]
                                 \centering
                                 \footnotesize
                                 \addtolength{\tabcolsep}{-0.2em}
  \begin{tabular}{llrrrrrr} \toprule
    & Model & \multicolumn{3}{c|}{ARCH(1)} &
                                             \multicolumn{3}{c}{GARCH(1,1)} \\
     \cmidrule(lr){3-5} \cmidrule(lr){6-8}
   Copula (form of density) & $p$ \textbar\ $(\epsilon_t)$ & Gauss & t4 &
                                                \multicolumn{1}{c|}{t2.5}
                                 & Gauss & t4 & t2.5 \\ \midrule
$c^t_\nu(u,v) = c^{|t|}_\nu(|2u-1|, |2v-1|) $  & 1 & 528 & 0 & 126 & 103 & 7 & 95 \\ 
 $\tfrac{1}{2} \left(c^t_{\rho,\nu}(u,v) +  c^t_{-\rho,\nu}(u,v)\right)$ & 2 & 530 & 2 & 128 & 101 & 9 & 97 \\ \midrule
  $\tfrac{1}{4}  \left(c_\theta^{\text{Cl}}(u,v) + c_\theta^{\text{Cl},90}(u,v) + c_\theta^{\text{Cl},180}(u,v) + c_\theta^{\text{Cl},270}(u,v)\right)$ & 1 & 1676 & 426 & 189 & 491 & 291 & 145 \\ 
  $\tfrac{1}{4}  \left(c_\theta^{\text{Gu}}(u,v) + c_\theta^{\text{Gu},90}(u,v) + c_\theta^{\text{Gu},180}(u,v) + c_\theta^{\text{Gu},270}(u,v)\right)$ & 1 & 913 & 150 & 182 & 228 & 120 & 156 \\ 
  $\tfrac{1}{4}  \left(c_\theta^{\text{J}}(u,v) + c_\theta^{\text{J},90}(u,v) + c_\theta^{\text{J},180}(u,v) + c_\theta^{\text{J},270}(u,v)\right)$ & 1 & 1982 & 593 & 276 & 666 & 431 & 232 \\ \midrule
  $c_\theta^{\text{Cl},180}(|2u-1|, |2v-1|)$ & 1 & 529 & 146 & 292 & 174 & 105 & 216 \\ 
  $c_\theta^{\text{Gu}}(|2u-1|, |2v-1|)$ & 1 & 1167 & 444 & 379 & 422 & 307 & 290 \\ 
  $c_\theta^{\text{J}}(|2u-1|, |2v-1|)$ & 1 & 650 & 50 & 145 & 144 & 65 & 131 \\ \midrule
  $\tfrac{1}{4} \left(c_{\theta,a_1,a_2}^{\text{Cl}}(u,v) + c_{\theta,a_1,a_2}^{\text{Cl},90}(u,v) + c_{\theta,a_1,a_2}^{\text{Cl},180}(u,v) + c_{\theta,a_1,a_2}^{\text{Cl},270}(u,v)\right)$  & 3 & 1170 & 430 & 102 & 385 & 288 & 70 \\ 
   $\tfrac{1}{4} \left(c_{\theta,a_1,a_2}^{\text{Gu}}(u,v) + c_{\theta,a_1,a_2}^{\text{Gu},90}(u,v) + c_{\theta,a_1,a_2}^{\text{Gu},180}(u,v) + c_{\theta,a_1,a_2}^{\text{Gu},270}(u,v)\right)$  & 3 & 854 & 154 & 66 & 199 & 124 & 64 \\ 
  $\tfrac{1}{4} \left(c_{\theta,a_1,a_2}^{\text{J}}(u,v) + c_{\theta,a_1,a_2}^{\text{J},90}(u,v) + c_{\theta,a_1,a_2}^{\text{J},180}(u,v) + c_{\theta,a_1,a_2}^{\text{J},270}(u,v)\right)$  & 3 & 1680 & 597 & 178 & 632 & 435 & 151 \\ \midrule
   $c_{\theta,a_1,a_2}^{\text{Cl},180}(|2u-1|, |2v-1|)$ & 3 & 46 & 147 & 156 & 89 & 98 & 119 \\ 
  $c_{\theta,a_1,a_2}^{\text{Gu}}(|2u-1|, |2v-1|)$& 3 & 598 & 448 & 221 & 330 & 303 & 159 \\ 
   $c_{\theta,a_1,a_2}^{\text{J}}(|2u-1|, |2v-1|)$ & 3 & 93 & 54 & 11 & 47 & 63 & 18 \\ 
   $c_{\nu,a_1,a_2}^{|t|}(|2u-1|, |2v-1|)$ & 3 & 0 & 3 & 0 & 0 & 0 & 0 \\ \bottomrule
    \end{tabular}
    \caption{Comparison of approximating copulas for the first-order
      copula of symmetric ARCH(1) and GARCH(1,1)
      processes. Estimates are obtained by fitting first-order D-vines
      using semiparametric method of~\cite{bib:chen-fan-06}.
      Each column shows the differences between AIC
    and the minimum value of AIC for that column; thus a 0
    indicates the model with lowest AIC. The superscripts Cl, Gu
    and J refer to Clayton, Gumbel and Joe copulas while 90, 180 and
    270 denote rotations.  Densities $c_{\theta,a_1,a_2}$ correspond
    to the Khoudraji construction in~\eqref{eq:6}.
    \label{tab:1}}
\end{table}

The results in Table~\ref{tab:1} relate to classical ARCH(1) and GARCH(1,1) models
with symmetric innovations. The innovation distributions considered are Gaussian, Student t4 and
Student t2.5, giving six different datasets in total. The spherical t
copula $C^t_\nu$
works quite well in all cases (despite its exchangeability) and is the
favoured model for the case of ARCH(1) with t4 innovations. It is
notable that the spherical t copula generally outperforms the mixture of two t
copulas; in the latter the estimate of $\rho$ is very close to zero
suggesting the model is poorly specified for this application.
For the other 
datasets the favoured copula has density of the form $c^*(|2u-1|, |2v-1|)$ where $C^*$ is the Khoudraji extension~\eqref{eq:6} of the
ast copula; the second-best copula is the model where $C^*$ is the Khoudraji extension
of either the Joe copula (four cases) or the survival Clayton copula (one case). Models constructed using explicit inverse
v-transformation are more effective than
mixture copulas
constructed according to approaches (2) and (3) of Section~\ref{sec:appr-mimick-garch}. While the Gumbel copula gives the best mixtures, the Joe and
survival Clayton copulas give better inverse-v-transformed models than
the Gumbel copula.

                              
\subsubsection{Models with asymmetries}

\newtext{When models contain asymmetric innovations, and potentially
  also leverage, 
we consider approximating GARCH pair copulas using the construction $c^*(\vtrans_1(u),
\vtrans_2(v))$ with asymmetric v-transforms $\vtrans(u;\delta,\kappa)$
taken from the
family~\eqref{eq:2_v}, which includes the
linear special case $\vtrans(u;\delta)$.}
\newtext{Section~\ref{sec:simul-study-asym} of the
Supplementary Material
contains extensive results for
ARCH(1) and GARCH(1,1)
models with skewed t innovations, with and without leverage effects.
The most important observation in these analyses is that, when two
different asymmetric v-transforms $\vtrans_1(u)$ and $\vtrans_2(v)$
are used, the Khoudraji
extension of the copula $C^*$ adds very little. This is a helpful insight from a computational point of view because
evaluating the density of the Khoudraji extension
requires, not only the density $c^*$ of the parent copula, but also
the function
$C^*$ itself, leading to much slower evaluation (particularly in the
ast case). The best models are of the two-parameter form
$c^*(\vtrans(u;\delta_1,\kappa_1),\vtrans(v;\delta_2,\kappa_2))$ with
$c^*$ taken to be $c^{|t|}_\nu$ or $c^{\text{J}}_\theta$. However
models of the form $c^*(\vtrans(u;\delta_1),\vtrans(v;\delta_2))$ with
two asymmetric linear v-transforms
already offer a large improvement on jointly symmetric models and offer
the advantages of greater parsimony and greater tractability, as we
discuss in Section~\ref{sec:pract-issu-constr}. Moreover, the analyses 
of
Section~\ref{sec:examples} will show that, in the presence of
asymmetry, models with two asymmetric linear
v-transforms tend to give better estimates of $C_1$ than mixtures of t
copulas or convex Gumbel copulas.
}


\section{Simplified D-vines for stochastic volatility}\label{sec:simplified-D-vines}

\subsection{Proposal for a simplified D-vine
  structure}\label{sec:pract-issu-constr}

\newtext{In this section we use insights from
  Sections~\ref{sec:arch-garch-copulas}
  and~\ref{sec:copula-dens-constr} to propose a class of
  simplified D-vines that behave like GARCH(1,1) processes and 
offer potential alternative models for financial
time series.}
\newtext{A guiding principle in constructing our models is parsimony.
  GARCH(1,1) models have relatively few parameters and an effective model
that incorporates both leverage and skewed 
innovations can be created with six parameters ($\alpha_0$, $\alpha_1$,
$\beta_1$, $\gamma_1$, $\xi$, $\lambda$). The first four
control the volatility equation and the final two describe the shape
and skew of the
innovation distribution; a seventh parameter $\mu$ could be added for
a non-zero mean. With a handful of parameters, the GARCH(1,1)
model is able to model the persistence of volatility observed in
financial data in a way
that low-order ARCH($p$) models are not.
To achieve a similar effect in D-vine models, we construct
sequences of pair copulas $(C_k)_{k \in \N}$ using three main ideas.}
\newtext{
  \begin{enumerate}
\item We choose pair copula densities of the form $c_k(u,v) = c^*(\vtrans(u;\delta_{1}),
  \vtrans(v;\delta_{2});\theta_k)$, based on two linear v-transforms
  and a copula family $C^*(u,v;\theta)$ satisfying $ c^*(0,0;\theta) < \infty$ which can model the full range of positive
  dependence between independence and comonotonicity for different
  values of the parameter $\theta$ in a parameter space $\Theta$;
  examples are the ast, Joe and survival Clayton copulas. This construction offers
  flexibility at the cost of only three
  parameters: $\theta_k$ controls the strength of dependence
  while $\delta_{1}$ and $\delta_{2}$ accommodate possible
  asymmetries. 
 \item We model the parameters $(\theta_k)_{k\in\N}$ as a
   parametric function of the lag $k$ so that $\theta_k = g(k,
   \bm{\vartheta})$ where $\bm{\vartheta}$ is a low-dimensional
   parameter vector. More details of this parameterisation are given below.
   \item We assume that the asymmetry parameters $\delta_{1}$ and
   $\delta_{2}$ are constant for all lags $k$. This strong
   assumption is
   based on the observation that these parameters are 
   equal to 0.5 at all lags in jointly symmetric models. It is also partly
   influenced by the GARCH(1,1) model where a single parameter $\gamma_1$
   controls leverage and a further parameter $\lambda$ controls
   asymmetry of innovations.
 \end{enumerate}
 }

 \newtext{There is another practical reason for choosing
 linear v-transforms. After making the simplifying assumption, the
copula
 density implied by the D-vine in equation~\eqref{eq:27} takes the form
                             \begin{equation}\label{eq:15}
                               c_{U_1,\ldots,U_n}(u_1,\ldots,u_n) = 
  \prod_{k=1}^{n-1}\prod_{j=k+1}^{n}
                  c_{k}(F_{U_{1}|\bm{U}_{[2:k]}}(u_{j-k}\mid
                  \bm{u}_{[j-k+1 : j-1]}
                  ),  F_{U_{k+1}|\bm{U}_{[2:k]}}(u_{j} \mid   \bm{u}_{[j-k+1 : j-1]})).
                \end{equation}
 It is well known that the 
 the evaluation of the conditional distribution functions
 $F_{U_{1}|\bm{U}_{[2:k]}}$ and $F_{U_{k+1}|\bm{U}_{[2:k]}}$ in~\eqref{eq:15} is straightforward if we have a method of computing the so-called
 $h$-functions of the copulas $C_k$, which are the partial derivatives
 $h_{k,1}(u,v) = \frac{\partial}{\partial u} C_k(u,v)$ and
 $h_{k,2}(u,v) = \frac{\partial}{\partial v} C_k(u,v)$; see, for
 example,~\cite{bib:aas-czado-frigessi-bakken-09} or~\cite{bib:smith-min-almeida-czado-10}. This is
 easy in the linear case but,
in the non-linear case, cumbersome
numerical integration seems in general to be necessary. More details of the formulas
for the linear case are found in Section~\ref{sec:calc-h-funct} of
the Supplementary Material while in Section~\ref{sec:link-vt-d} we show
that, in the special case of two identical linear v-transforms,
D-vines with inverse-v-transformed copulas also belong to
the family of vt-D-vines proposed in~\cite{bib:bladt-mcneil-21}.}

\newtext{The relationship $\theta_k = g(k,\bm{\vartheta})$ between the
  parameters $\theta_k$ and the lag $k$ is determined
  implicitly using
  a method developed in~\cite{bib:bladt-mcneil-22} and~\cite{bib:bladt-dias-han-mcneil-25}. Let
  $\tau^*: \Theta \to [0,1)$ be the invertible function that maps the parameter value
  of the copula $C^*(u,v;\theta)$ into the Kendall's tau measure
  of association and write $(\tau^{*})^{-1}$ for its inverse.
  Let $(\omega_k(\bm{\vartheta}))_{k\in\N}$ denote the
  partial autocorrelation (pacf) of
  a classical Gaussian ARMA($p$,$q$) process, where $\bm{\vartheta}$ denotes the vector
  of $p+q$ parameters and where we assume that the parameters are
  constrained to ensure
  $\omega_k(\bm{\vartheta}) \in [0,1)$ for all $k \in \N$. We determine
  the sequence of parameters $(\theta_k)_{k\in\N}$ by setting
$\theta_k= (\tau^{*})^{-1}\left(2\pi^{-1}\arcsin(\omega_k(\bm{\vartheta}))\right)$ for all $k$. This means that
we equate the Kendall's tau values of the copula sequence
$(C^*(u,v;\theta_k))_{k\in \N}$ with the
Kendall's tau values of the sequence of pair copulas that characterise
the Gaussian ARMA($p$,$q$)
model. In practice we favour the ARMA(1,1) model based on the fact
that
the autocorrelation function of the squared values of the
classical GARCH(1,1) process decays like an ARMA(1,1) process, as
noted in Section~\ref{sec:class-garch11-plot}. We also consider AR($p$)
models to faciliate comparison with D-vine($p$) models in the
literature; these lead to the estimation of $p$ parameters $\theta_1,\ldots,\theta_p$.}

\subsection{Examples}\label{sec:examples}

\newtext{To illustrate the potential of the D-vine models proposed in
Section~\ref{sec:pract-issu-constr}, we provide four examples, three
using simulated data from the GARCH(1,1)-type processes documented in
Table~\ref{tab:parameters}
and one using real daily return data for the USD/AUD exchange rate, as previously analysed
in~\cite{bib:loaiza-maya-et-al-18}.
The simulated datasets consist of 10000 values from the GARCH(1,1)
model with (i) t4 innovations, (ii) skewed t4 innovations and (iii)
skewed t4 innovations and leverage.
The exchange-rate returns consist of 3669 values from 2 January
2001 to 7 August 2015.}

\newtext{Models are fitted using the
generalization of the two-stage method of~\cite{bib:chen-fan-06} to
higher-order D-vines. For the simulated data, the marginal
distribution function is estimated by the empirical distribution
function while, for the real data, we use the same adaptive kernel
density estimator as~\cite{bib:loaiza-maya-et-al-18}
to facilitate comparison of results. The maximum-likelihood fit of the parametric copula model to the
pseudo-copula observations can be accomplished in a single step, rather
than by using the common stepwise procedure where the sequence of
$C_1,\ldots,C_p$ defining a D-vine($p$) model is estimated in $p$ successive steps.}

\newtext{To all datasets we fit D-vine(1) and D-vine(5) models based on t
copulas (T1, T5) and convex Gumbel mixtures (B1, B5), as used
in~\cite{bib:loaiza-maya-et-al-18}.
For the new models of Section~\ref{sec:pract-issu-constr}, we use
inverse-v-transformed Joe, survival Clayton and ast copulas and
parameterizations based on
the partial acfs of Gaussian AR(1), AR(5)
and ARMA(1,1) processes. The copula sequence in the ARMA(1,1) case is
truncated at 40 for the real data and 20 for the simulated data, since
persistence of volatility is weaker for the latter. For the empirical data only, we add D-vine(1)
and D-vine(5) models based on
mixtures of two t copulas (A1, A5) to facilitate further comparison
with models in~\cite{bib:loaiza-maya-et-al-18}. For the simulated data,
we add a D-vine(2)
model with two t copulas, as suggested by~\cite{bib:zhao-shi-zhang-22}
for GARCH approximation. Since the different models have very
different numbers of parameters, we compare models using both AIC and
BIC, where the latter makes greater adjustments for overfitting.}

\newtext{For all four datasets the best model, as selected by both the AIC
  and BIC information criteria, is ast-ARMA(1,1) models and we
  show estimated parameter values and standard errors in
  Table~\ref{tab:D-vines}.
It may be noted that the estimates of
$\delta_1$ and $\delta_2$ are both close to 0.5 for the symmetric GARCH
model, but more offset in the other cases.
The estimates of the AR and MA parameters ($\phi$ and $\psi$) are
larger in absolute value for the real data than for the GARCH
simulations, implying a much slower decay of the partial acf and greater persistence in
volatility.
We also give the parameters $\nu_1$ and $\nu_2$ of the ast copulas at
the first two lags as well as the corresponding values of Kendall's tau,
$\tau_1$ and $\tau_2$.
}

\begin{table}[htb]
  \centering
  \footnotesize
\begin{tabular}{lcccccccc}
  \toprule
 Dataset & $\phi$ & $\psi$ & $\delta_1$ & $\delta_2$ &
                                                                 $\nu_1$ & $\nu_2$ & $\tau_1$ & $\tau_2$ \\ 
  \midrule
 GARCH(1,1), t4 innovations & 0.780 & -0.614 & 0.510 & 0.515 & 2.96 & 5.07 & 0.132
                                                                                              & 0.079 \\ 
   & (0.012) & (0.016) & (0.017) & (0.016) &  &  &  &  \\ 
 GARCH(1,1), skew t4 innovations& 0.786 & -0.614 & 0.493 & 0.427 & 2.82 & 4.85 &
                                                                     0.139 & 0.082\\ 
 & (0.012) & (0.017) & (0.016) & (0.016) &  &  &  &  \\ 
 GARCH(1,1), skew t4 + leverage & 0.790 & -0.607 & 0.644 & 0.434 & 2.60 & 4.56 &
                                                                     0.150 & 0.088 \\ 
         & (0.013) & (0.018) & (0.015) & (0.014) &  &  &  &  \\
  USD/AUD daily exchange-rate returns & 0.982 & -0.934 & 0.528 & 0.446
                                                     & 5.82 &6.59 & 0.069& 0.061\\
  & (0.0029) & (0.0073) & (0.024) & (0.023) & & & & \\
   \bottomrule
\end{tabular}
\caption{Parameter estimates and (standard errors) for D-vines based
  on ast copula and partial acf of ARMA(1,1) model fitted to four
  datasets. Models have four primary parameters, $\phi$ and $\psi$
  from the ARMA(1,1) partial acf and $\delta_1$ and $\delta_2$ from
  the linear v-transforms; $\nu_1$ and $\nu_2$ are implied parameters of the ast copulas at
the first two lags and $\tau_1$ and $\tau_2$ are the corresponding values of Kendall's tau.
\label{tab:D-vines}}
\end{table}





\newtext{The full model comparisons for the USD/AUD data are found in
  Table~\ref{table:empirical} and we give more commentary on these
  results since they are broadly representative of all four examples; corresponding tables for the three
  simulated datasets are given in Section~\ref{sec:detailed-results} of the Supplementary Material.
Comparing D-vine(1) models, the copulas based on
inverse-v-transformation of the Joe and ast copulas lead to higher
log-likelihoods and lower AIC and BIC values than the convex Gumbel mixture
model B1. This suggests the
fundamental shape of these copulas is slightly better than the
mixture copulas for modelling $C_1$. Comparing D-vine(5)
models, B5 gives the highest log-likelihood value but is a very
expensive model to fit in terms of the number of parameters and the time
required to maximize the
likelihood\footnote{\cite{bib:loaiza-maya-et-al-18} fit this model
  using MCMC in the Bayesian paradigm and we use their posterior mean
  estimates as starting values in maximum likelihood estimation.}. The D-vine(5)
models based on inverse-v-transformed copulas have less than a third
of the parameters, are quicker to fit and lead to lower values for the
information criteria. However, our favoured D-vine(40) models based on the partial acf of
an ARMA(1,1) model lead to the largest log-likelihood values and the
lowest AIC and BIC values at the cost of only 4 parameters. }

\newtext{\cite{bib:loaiza-maya-et-al-18} show that their favoured
  models (B5 and T5) lead to one-step-ahead forecasts of conditional
  quantiles (value-at-risk) that are more accurate than the ones
  produced by classical symmetric GARCH models and are not
  rejected by Christoffersen's test of conditional coverage at the
  99\% level~\citep{bib:christoffersen-98}. This is also true of our
  favoured models (Joe-AR(5), ast-AR(5), Clayton180-ARMA(1,1) and
  ast-ARMA(1,1)); results are documented in
  Section~\ref{sec:analys-var-exce} of the Supplementary Material.
  }

\begin{table}[htb]
\centering
\begin{tabular}{lrrrrr}
\toprule
  Model & D-vine order & No.\ pars & Loglik & AIC & BIC\\
  \midrule
  T1 & 1 & 2 & 36.72 & -69.43 & -57.01 \\
  A1 & 1 & 5 & 36.71 & -63.42 & -32.38 \\
  B1 & 1 & 5 & 38.70 & -67.40 & -36.36 \\
  Joe-AR(1) & 1 & 3 & 41.17 & -76.33 & -57.71 \\
  Clayton180-AR(1) & 1 & 3 & 38.26 & -70.51 & -51.89\\
  ast-AR(1) & 1 & 3 & 41.13 & -76.27 & -57.64 \\
      \midrule
  T5 & 5 & 10 & 153.04 & -286.08 & -224.01 \\
  A5 & 5 & 25 & 155.63 & -261.26 & -106.07 \\
  B5 & 5 & 25 & 163.09 & -276.17 & -120.98 \\
  Joe-AR(5) & 5 & 7 & 154.13 & -294.26 & -250.81 \\
  Clayton180-AR(5) & 5 & 7 & 151.32 & -288.63 & -245.18\\
  ast-AR(5) & 5 & 7 & 155.72 & -297.45 & -253.99 \\
  \midrule
  Joe-ARMA(1,1)        & 40 & 4 & 247.21 & -486.42 & -461.59 \\
  Clayton180-ARMA(1,1) & 40 & 4 & 255.29 & -502.59 & -477.76 \\
  ast-ARMA(1,1)        & 40 & 4 & 261.91 & -515.83 & -491.00 \\
\bottomrule
\end{tabular}
\caption{Comparison of D-vine models 
  in~\cite{bib:loaiza-maya-et-al-18} and the D-vine models of
  Section~\ref{sec:simplified-D-vines} fitted to return data for the
  USD/AUD exchange rate for the period 2 January 2001 to 7 August
  2015; the marginal distribution is estimated nonparametrically using
  adaptive kernel density estimation. T1, A1
  and B1 are D-vine(1) models based, respectively, on t copulas, mixed t copulas and
  mixed convex Gumbel copulas as used
  in~\cite{bib:loaiza-maya-et-al-18}, while T5, A5 and B5 are
  D-vine(5) extensions. The models following the construction in
  Section~\ref{sec:pract-issu-constr} are described by the name of the copula $C^*$ and the
  name of the AR or ARMA model underlying the parameterisation $\theta_k
  =g(\bm{\vartheta})$; the copula sequence in the ARMA(1,1) case is
  truncated at order 40. \label{table:empirical}}
\end{table}

\section{Conclusion}\label{sec:conclusion}

\newtext{This paper has asked three questions. First, what are the bivariate
copulas that underlie the serial dependencies and conditional
dependencies of pairs of variables $(X_t, X_{t+k})$ in first-order
GARCH(1,1) models? Second, since
the true pair copulas and conditional copulas of GARCH-type
processes are inaccessible for practical use,
can they be approximated by simpler
constructions? Third, can these ideas help us to
construct simplified D-vines that behave like GARCH-type models, but
offer more flexible models for real data showing stochastic volatility?}

With regard to the first question we find a number of
symmetries in the dependence structures of the models with symmetric
innovation distributions, even when leverage effects are introduced,
but these symmetries disappear in the presence of asymmetric
innovations, which are often necessary to model
real-world data in practice. What can always be said is that the
underlying bivariate copulas are cross-shaped, in the sense that there
is a tendency for points to accumulate in all fours corners of the
unit square, and that this tendency increases at shorter lags as the size of the ARCH
effect increases and persists at longer
lags as the aggregate size of the ARCH
and GARCH effects increases.

Motivated by the recognition that v-transforms are implicitly
embedded in the copulas of symmetric GARCH models, we propose
a construction to approximate GARCH pair copulas based on v-transforms and 
positive-dependence copulas with
strong dependence in the upper joint tail.
Moreover, we identify a new copula - the absolute
spherical t copula - which offers a promising basis
for the construction.

\newtext{
  With regard to the final question, we believe that the simplified D-vine models 
  proposed in Section~\ref{sec:pract-issu-constr} offer 
  effective descriptions of GARCH-like behaviour and have 
  potential for modelling volatile financial
  data. Their advantages come not only from the particular
  choice of pair copulas, but also from a method of
  parameterization, using the partial acf of Gaussian ARMA processes,
  which gives models with a similar degree of parsimony to GARCH
  processes. In combination with the flexibility to choose an
  appropriate marginal distribution, these models could be extremely
  effective when data depart from the power-lawed marginal behaviour
  imposed by a GARCH process, as discussed in~\cite{bib:bladt-mcneil-21}.
  However, the proposed models also need more extensive empirical
  testing. It is important to conduct experiments with a wider
  variety of financial return
data, such as stock-price, index, commodity-price and
volatility-index returns. This is the subject of future work.
}

Another issue requiring further research is the construction of
D-vines to model ARMA-GARCH behaviour.
Real-world financial return series sometimes show serial correlation
in the raw data as well as in the squared or absolute value
series and this is usually modelled by combining ARMA and GARCH
effects. In Section~\ref{sec:ar1-model-with} of the Supplementary Material we
consider
an AR-ARCH model and suggest that mixtures
of standard positive-dependence copulas and the pair copulas of this paper could be
successful in addressing combined ARMA and GARCH effects.

\newpage

\begin{center}
  {\Large \textbf{Supplementary Material}}
  \end{center}

  \appendix
  \makeatletter
\renewcommand \thesection{S\@arabic\c@section}
\setcounter{table}{0}
\renewcommand\thetable{S\@arabic\c@table}
\setcounter{figure}{0}
\renewcommand \thefigure{S\@arabic\c@figure}
 \setcounter{equation}{0}
\renewcommand \theequation{S\@arabic\c@equation}
 \setcounter{proposition}{0}
\renewcommand \theproposition{S\@arabic\c@proposition}
\makeatother

\section{Correlation and rank correlation for h-symmetric distributions}\label{sec:corr-rank-corr}
\citet{bib:nelsen-99} notes that jointly symmetric distributions have
zero correlation (when finite second moments permit the measure to be
defined). Perhaps surprisingly, h-symmetry on its own is a sufficient condition for an
absence of both correlation and rank correlation; the same is true of
v-symmetry.
\begin{proposition}
If a random vector $(Y,Z)^\top$ has finite second moments and is
either h-symmetric or v-symmetric then the Pearson
correlation $\rho(Y,Z)$ is zero.
\end{proposition}
\begin{proof}
 We give the proof in the case of h-symmetry, the method for
 v-symmetry being identical. If $(Y,Z)^\top$ is h-symmetric about $b$
 then $\E(Z) = b$ and we can write
 $\cov(Y,Z) = \E(YZ) -\E(Y)\E(Z) = \E(Y(Z-b)) + \E(Yb) -\E(Y)b = \E(Y(Z-b))$.
 However, h-symmetry implies
$\E(Y(Z-b)) = \E(Y(b-Z)) = -\E(Y(Z-b))$ implying that $\E(Y(Z-b))=0$ from
which the result follows.
  \end{proof}

\begin{proposition}
If a random vector $(Y,Z)^\top$ has continuous marginal distributions
and is either h-symmetric or v-symmetric, then the Spearman rank
correlation $\rho_S(Y,Z)$ and the Kendall rank correlation $\tau(Y,Z)$
are both zero.
\end{proposition}
\begin{proof}
 Again it suffices to give the proof in the case of h-symmetry.
  Let the unique copula of $(Y,Z)^\top$ be denoted $C$.
For Spearman's rank correlation we need to calculate $\rho_S(Y,Z) = 12\int_0^1\int_0^1
(C(u,v) - uv) \rd u \rd v$. From the identity $C(u,v) = u - C(u,1-v)$
we infer that
\begin{eqnarray*}
 \int_0^1\int_0^1
C(u,v)  \rd u \rd v &=&  \int_0^1\int_0^1
u  \rd u \rd v - \int_0^1\int_0^1
C(u,1-v)  \rd u \rd v = \frac{1}{2} - \int_0^1\int_0^1
C(u,y)  \rd u \rd y
\end{eqnarray*}
where we have made the change of variable $y = 1-v$. We must have
$ \int_0^1\int_0^1
C(u,v)  \rd u \rd v = \tfrac{1}{4}$ and then simple calculations give
$\rho_S(Y,Z) = 0$.

In the case of Kendall's rank correlation we need
to compute $\tau(Y,Z) = 4\E (C(U,V)) -1$ where $(U,V)^\top$ is a pair
of variables with distribution function $C$. The
identity $C(u,v) = u - C(u,1-v)$ implies that $(U,V)^\top \eqdis
(U,1-V)^\top$; indeed $(U,V-0.5)^\top \eqdis
(U,0.5 - V)^\top$ since $(U,V)^\top$ is h-symmetric about 0.5. Hence we infer that
\begin{eqnarray*}
 \E (C(U,V)) & = & \E(U) - \E(C(U,1-V)) = \frac{1}{2} - \E(C(U,V))
\end{eqnarray*}
which implies that $\E(C(U,V)) = \tfrac{1}{4}$ and
$\tau(Y,Z) =0$.
\end{proof}

\section{ More illustrations of joint densities and copulas}\label{sec:pictures}

\begin{figure}[h!]
\centering
\includegraphics[width=8cm,height=8cm]{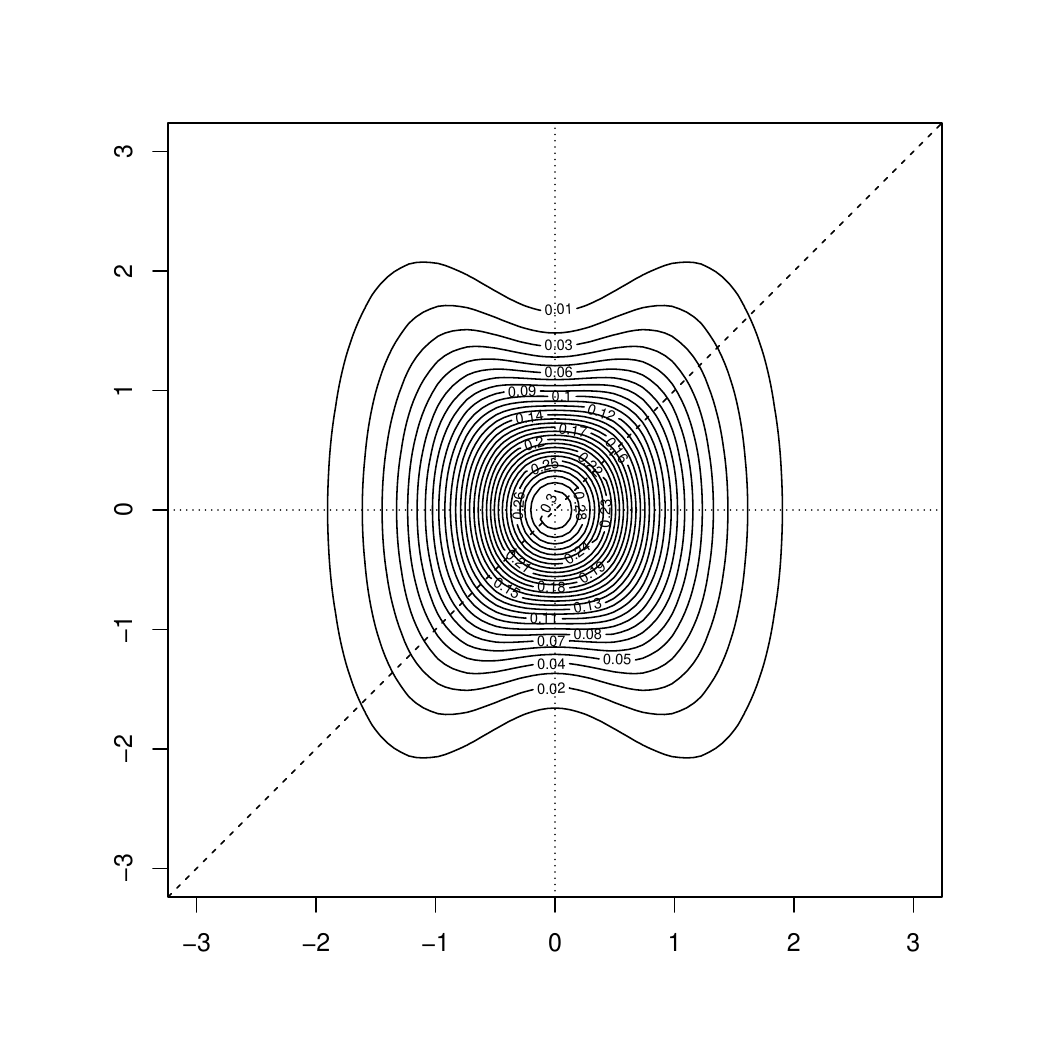}
\includegraphics[width=8cm,height=8cm]{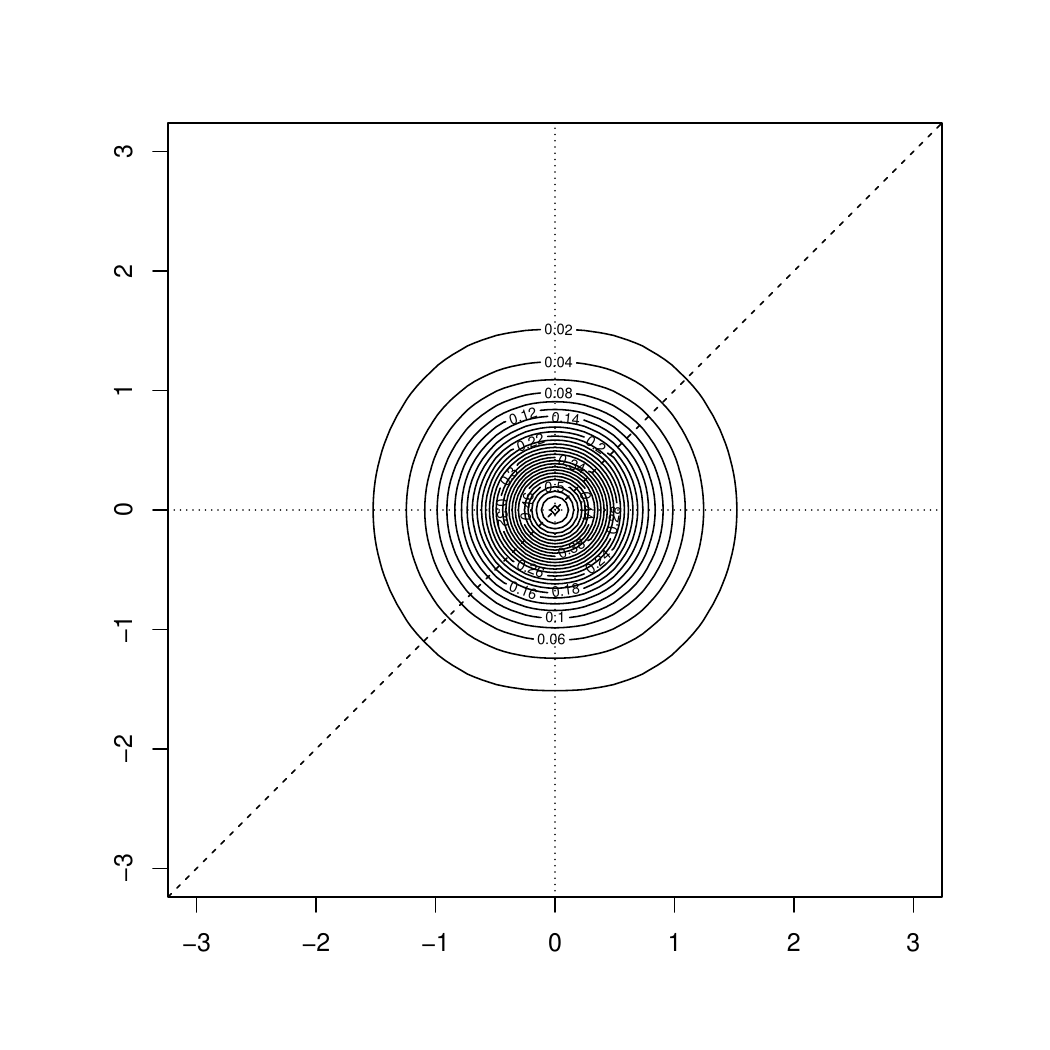}\\
\includegraphics[width=8cm,height=8cm]{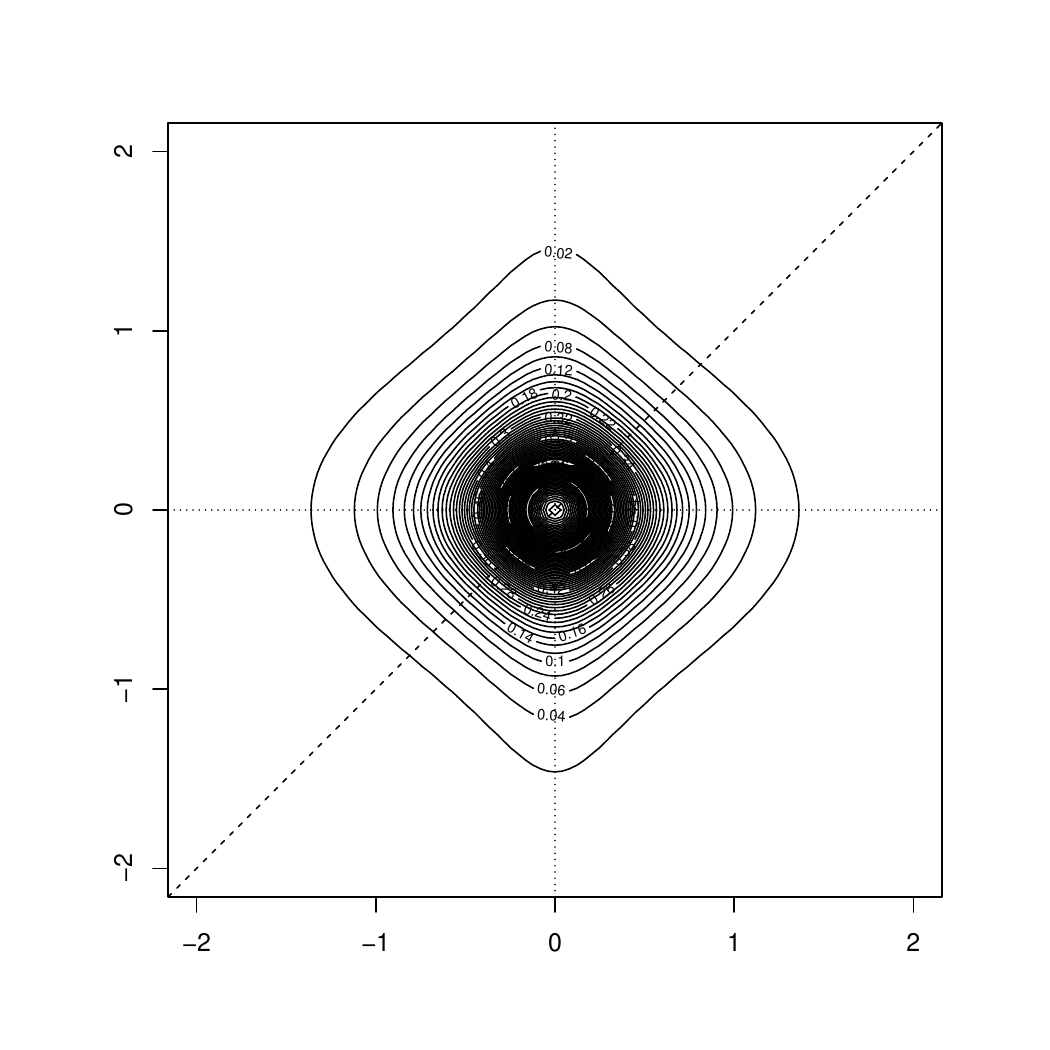}
\includegraphics[width=8cm,height=8cm]{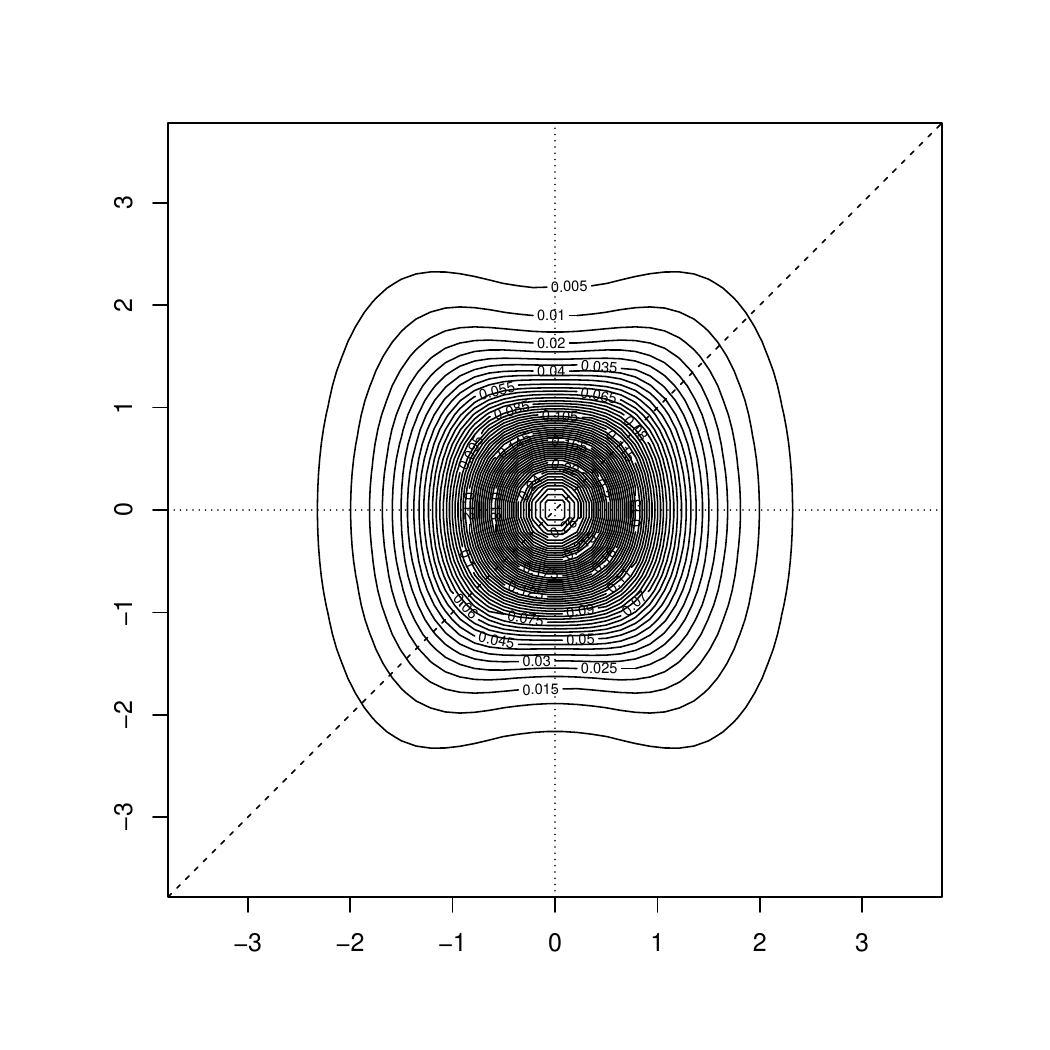}
 \caption{Contour plots of $f_{X_1,X_2}(x_1,x_2)$ for symmetric 
   GARCH(1,1)-type processes; see Table~\ref{tab:parameters} for all
   parameter values. Top left, top right, bottom left: classical ARCH(1)
   processes with,
   respectively, Gaussian, Student t4 and Student t2.5 innovations. Bottom right: classical
   GARCH(1,1) process with Gaussian innovations.
}
\label{fig:0}
\end{figure}

Figure~\ref{fig:0} shows contour plots of $f_{X_1,X_2}(x_1,x_2)$ for a
number of symmetric 
   GARCH(1,1)-type processes as described in
   Sections~\ref{sec:class-arch1-plot}--\ref{sec:class-garch11-plot}; these pictures are counterparts to the
   copula pictures in Figure~\ref{fig:1} and show the joint
   symmetry of the distributions. It is notable that the contours in
   the top right
   picture for an ARCH(1) model with Student t4 innovations look
   almost circular, showing the distribution is very close to
   spherical; the corresponding copula density plot in Figure~\ref{fig:1}
   suggests the copula is almost exchangeable.

In the left panel of Figure~\ref{fig:OS} we
  show a contour plot of the absolute value copula (the copula of
  $(|X_1|,|X_2|)^\top$) for the ARCH(1)
  model with Gaussian innovations described in Section~\ref{sec:class-arch1-plot}.
  This can be thought of as a blown-up version of
  the top right quadrant of the copula density for $(X_1,X_2)^\top$ and
clearly shows 
non-exchangeability and stronger dependence in the
upper tail.
Note that this picture also illustrates the density of the volatility
  copula of ARCH(1), that is the copula of
  $(\sigma_1,\sigma_2)^\top$; this follows because $\sigma_t$
  is a strictly increasing function of $|X_{t-1}|$ for all $t$ and copulas are
  invariant under strictly increasing transformations of
  variables. For the same reason, it is also the copula density for the
variables $(X_{1}^2,X_{2}^2)^\top$.

\begin{figure}[h!]
  \centering
\includegraphics[width=8cm,height=8cm]{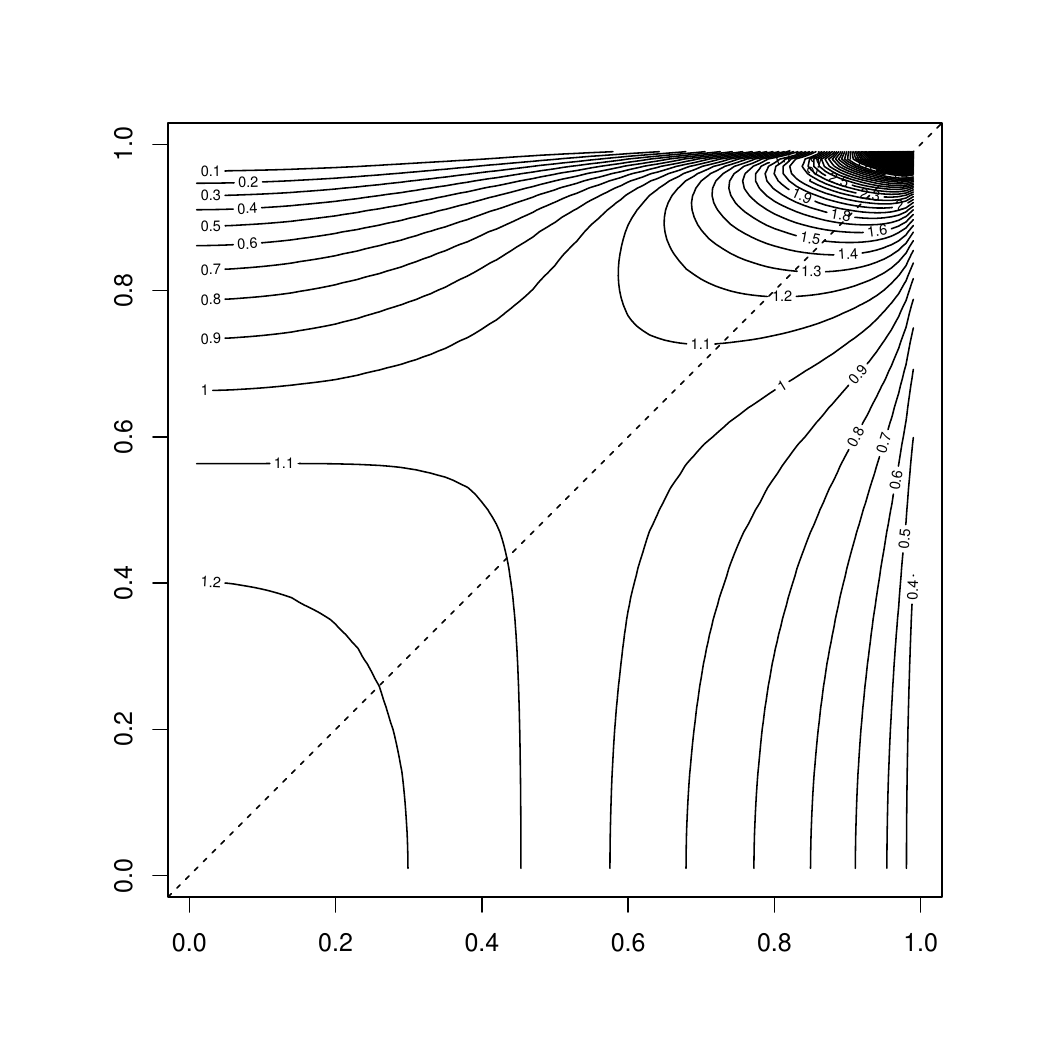}
\includegraphics[width=8cm,height=8cm]{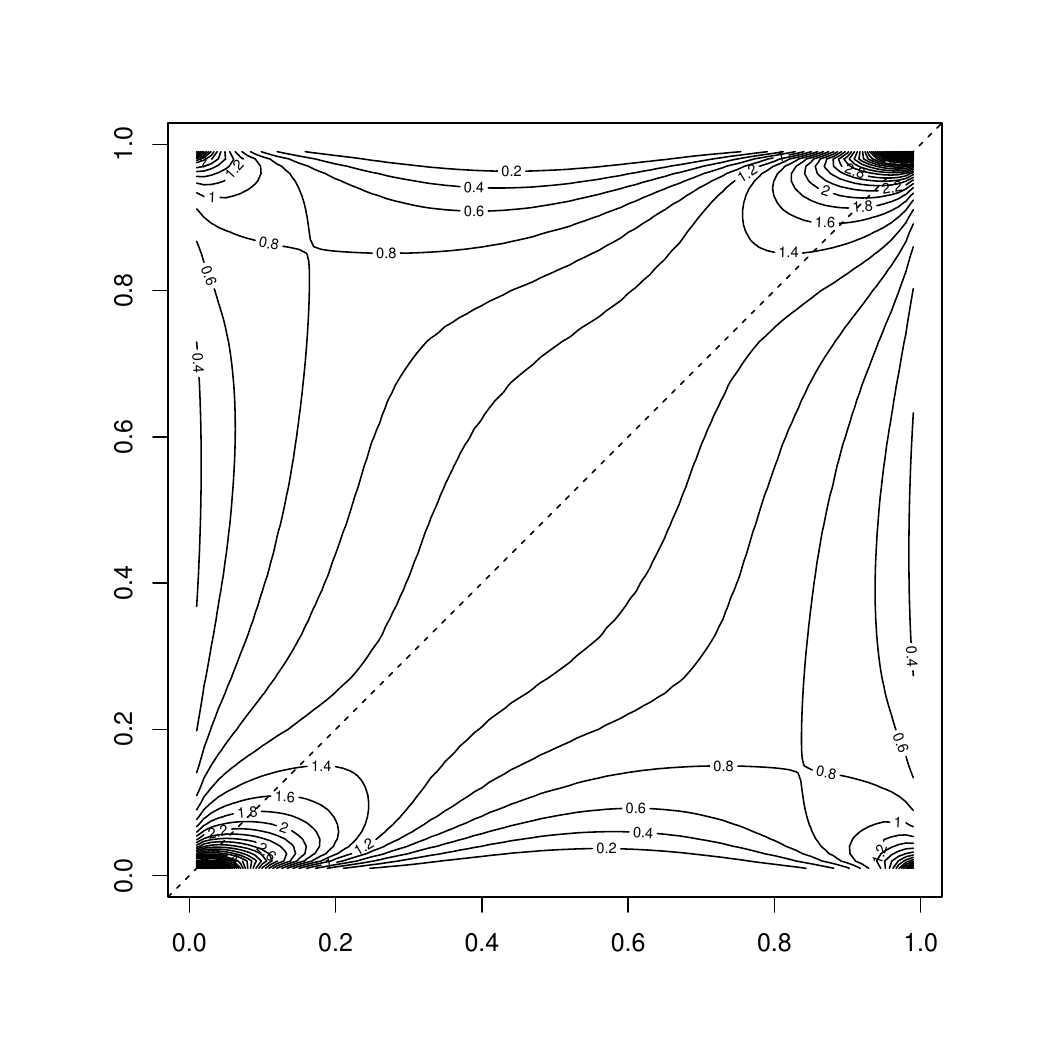}
\caption{Left: copula density for copula
of the absolute values for classical ARCH(1) copula with Gaussian
innovations; see Section~\ref{sec:class-arch1-plot} for parameters.
Right: AR(1)-ARCH(1) model with Gaussian innovations and
parameters $\phi = 0.3$, $\alpha_0=0.4$ and $\alpha_1=0.6$.
}
\label{fig:OS}
\end{figure}



\section{General proof of Proposition~\ref{prop:symmetry2}}

We complete the proof of Proposition~\ref{prop:symmetry2}. 
The case $k=2$ has been proved so let $k \geq 3$. Suppose we write $s_2(x,s) = \sigma(x,s)$, $s_3(z_1,x,s) = \sigma(z_1,s_2(x,s))$ and then for
$j=4,\ldots,k+1$ we define recursive functions $s_j(z_{j-2},\ldots,z_1,x,s) =
\sigma(z_{j-2}, s_{j-1}(z_{j-3},\ldots,z_1,x,s))$. This permits us to
write the joint density as
\begin{multline*}
    f_{X_1,\ldots,X_{k+1}}(x,z_1,\ldots,z_{k-1},y) 
                         = \int_0^\infty
                           f_{X_{k+1} \mid \sigma_{k+1}}(y \mid
                           s_{k+1}(z_{k-1},\ldots,z_1,x,s))\\
\left(\prod_{j=3}^k f_{X_j \mid \sigma_j}(z_{j-1} \mid
  s_j(z_{j-2},\ldots,z_1,x,s))\right)     f_{X_2
                             \mid \sigma_2}(z_1 \mid s_2(x,s))    f_{X_1 \mid \sigma_1}(x
\mid s ) f_\sigma(s)
\rd s .
\end{multline*}
This satisfies
$ f_{X_1,\ldots,X_{k+1}}(x,\ldots,y) =
f_{X_1,\ldots,X_{k+1}}(x,\ldots,-y)$ when the innovation
distribution is symmetric. Moreover, in the
symmetric model, we have  $ f_{X_1,\ldots,X_{k+1}}(x,\ldots,y) =
f_{X_1,\ldots,X_{k+1}}(-x,\ldots,y)$ since all of the
functions $s_2,\ldots, s_{k+1}$ are even functions of $x$. These properties also
hold for the conditional density
\begin{displaymath}
  f_{X_1,X_{k+1} \mid \bm{X}_{[2:k]}}(x,y \mid z_1,\ldots,z_{k-1}) =
  f_{X_1,\ldots,X_{k+1}}(x,z_1,\ldots,z_{k-1},y) /    f_{X_2,\ldots,X_{k}}(z_1,\ldots,z_{k-1}) .
\end{displaymath}
To infer that the
assertions regarding symmetry of the conditional copula hold we note that
h-symmetry of $f_{X_1,X_{k+1} \mid \bm{X}_{[2:k]}}$ implies that
$F_{X_{k+1} \mid \bm{X}_{[2:k]}}$ is symmetric and joint symmetry of
$f_{X_1,X_{k+1} \mid \bm{X}_{[2:k]}}$ also implies that
$F_{X_{1} \mid \bm{X}_{[2:k]}}$ is symmetric; the result then follows easily
 from equation~\eqref{eq:2}.

  Since the joint density of $(X_1,X_{k+1})^\top$ can be obtained as
  \begin{displaymath}
    f_{X_1,X_{k+1}}(x,y) = \int_{-\infty}^\infty \cdots
    \int_{-\infty}^\infty
      f_{X_1,\ldots,X_{k+1}}(x,z_1,\ldots,z_{k-1},y) \rd z_1\cdots \rd z_{k-1}
    \end{displaymath}
    it follows by similar logic and application of Lemma~\ref{lemma:symm-copulas} 
    that the copula $C_{(k)}$ of the unconditional distribution 
    of $(X_1,X_{k+1})^\top$ shares the same symmetry properties as the
    conditional copula $C_k$.

\section{The absolute spherical t copula}\label{sec:ast-copula}
The absolute spherical t copula has density given by
\begin{displaymath}
    c^{|t|}_\nu(u,v) = c^t_\nu\left(\frac{1-u}{2},\frac{1-v}{2}\right) =
  c^t_\nu\left(\frac{1-u}{2},\frac{1+ v}{2}\right) =
  c^t_\nu\left(\frac{1+u}{2},\frac{1-v}{2}\right) = c^t_\nu\left(\frac{1+u}{2},\frac{1+v}{2}\right).
\end{displaymath}
It has an asymptote at $(1,1)$ but takes a finite value at $(0,0)$
given by
\begin{displaymath}
   c^{|t|}_\nu(0,0) = \frac{\Gamma\left( \frac{\nu+2}{2} \right) \Gamma\left(\frac{\nu}{2}
     \right)}{\Gamma\left( \frac{\nu+1}{2}\right)^2}.
 \end{displaymath}
 This behaviour facilitates the maximum likelihood estimation of
 fulcrum parameters in
 models that incorporate v-transforms.
 
Integrating any of the four equalities for $c^{|t|}_\nu$ and observing the boundary
conditions for a copula gives us an expression for 
$C^{|t|}_\nu$. For example, integrating the final equality gives
\begin{equation}\label{eq:33}
  C^{|t|}_\nu(u,v) = 4C_\nu^t\left(\frac{1+u}{2}, \frac{1+v}{2}\right) -u -
  v -1.
\end{equation}
The h-functions for the copula may be calculated to be
\begin{equation}\label{eq:h-for-ast}
\begin{aligned}
  h^{|t|}_{\nu,1}(u,v) & = \frac{\partial}{\partial u}
                         C^{|t|}_\nu(u,v)  = 2 t_{\nu +1}\left(\sqrt{\frac{\nu+1}{\nu
        + (t_\nu^{-1}(\frac{1+u}{2}))^2}}
                         t_\nu^{-1}\left(\frac{1+v}{2}\right) \right)
                         -1 \\
   h^{|t|}_{\nu,2}(u,v) & = \frac{\partial}{\partial v}
                         C^{|t|}_\nu(u,v) = 2 t_{\nu +1}\left(\sqrt{\frac{\nu+1}{\nu
        + (t_\nu^{-1}(\frac{1+v}{2}))^2}}
                         t_\nu^{-1}\left(\frac{1+u}{2}\right) \right)
                         -1 
                       \end{aligned}
                       \end{equation}
where $t_\nu$ denotes the distribution function of a Student t
distribution with $\nu$ degrees of freedom.

For calculation purposes,
the h-functions are easy to evaluate for all $\nu > 0$,
since the function $t_\nu$ is typically available in statistical
software. This is not however the case for the copula $C^{|t|}_\nu$
itself, since available software for evaluating the t copula
$C^t_{\nu,\rho}$ is currently restricted. The \texttt{R} package
\texttt{copula} is restricted to integer $\nu$ while the package
\texttt{rvinecopulib} is confined to $\nu \geq 2$. A one-dimensional numerical integral seems unavoidable for
evaluating the copula for general $\nu$ and we solve this problem by
integrating the h-functions in~\eqref{eq:h-for-ast}.

From expression~\eqref{eq:33} we can infer that the absolute t copula interpolates between independence (as $\nu \to \infty$) and
comonotonicity (as $\nu \to 0$). For the latter case this follows
because the spherical t copula $C_\nu^t$ converges to $0.5 M(u,v) + 0.5 W(u,v)$,
an equal mixture of the comonotonicity and countermonotonicity copulas, as $\nu \to
0$~\citep{bib:mcneil-neslehova-smith-22}.

In view of the
limiting behaviour of the copula, the Kendall's tau value $\tau(\nu)$
of $C^{|t|}_\nu$ must range in the interval $(0,1)$ as $\nu$ ranges in
the interval $(0,\infty)$. We can find values by numerical integration
using the formula
\begin{equation}\label{eq:34}
  \tau(\nu) =
4\int_0^1 \int_0^1 C^{|t|}_\nu(u,v) c^{|t|}_\nu(u,v) \rd u \rd v -1 =
  64 \int_0^{0.5} \int_0^{0.5} C_\nu^t(u,v) c_\nu^t(u,v) \rd u,
  \rd v -1.
\end{equation}
For example, for $\nu \in \{4,2,1,0.5\}$ the corresponding values are
$\{0.099, 0.189, 0.333, 0.515 \}$.
The absolute spherical t copula has upper tail dependence (inherited from the t
copula). However we can easily show that it is asymptotically
independent in the lower tail.

\section{Simulation study of
  the approximation of $C_1$}\label{sec:simstudy}
In the simulation study, data are generated from ARCH(1) and GARCH(1,1)-type
processes using the parameter values summarized in
Table~\ref{tab:parameters} and different innovation distributions, as
recorded in the legend to the table.
Analyses are based on
simulated samples of 50000 points and the first-order 
copula $C_1$ is estimated by fitting a first-order D-vine
to the data using the two-stage method developed by~\cite{bib:chen-fan-06} in
which the marginal distribution of the process is estimated  by the
rescaled empirical distribution function.
Models are
compared according to their AIC values. The results of the jointly
symmetric case are documented in Section~\ref{sec:jointly-symm-case}. Further experiments are documented here.

\subsection{Simulation study of h-symmetric case with leverage}\label{sec:simul-study-h-sym}

\begin{table}[htbp]
  \centering
  \begin{tabular}{llrrrrrr} \toprule
    & Model & \multicolumn{3}{c|}{ARCH(1)} &
                                             \multicolumn{3}{c}{GARCH(1,1)} \\
     \cmidrule(lr){3-5} \cmidrule(lr){6-8}
   Copula (form of density) & $p$ \textbar\ $(\epsilon_t)$ & Gauss & t4 &
                                                \multicolumn{1}{c|}{t2.5}
                                 & Gauss & t4 & t2.5 \\ \midrule
$ c^{|t|}_\nu(|2u-1|, |2v-1|) $ & 1 & 1068 & 224 & 223 & 720 & 264 & 187 \\ 
  $c_\theta^{\text{J}}(|2u-1|, |2v-1|)$& 1 & 1202 & 267 & 238 & 764 & 321 & 214 \\ 
   $c_\theta^{\text{Cl},180}(|2u-1|, |2v-1|)$ & 1 & 1097 & 358 & 364 & 843 & 386 & 313 \\ \midrule
  $ c^{|t|}_\nu(\vtrans(u; 0.5,\kappa), |2v-1|) $ & 2 & 346 & 0 & 132 & 60 & 2 & 82 \\ 
 $c_\theta^{\text{J}}(\vtrans(u; 0.5,\kappa), |2v-1|)$  & 2 & 462 & 48 & 153 & 91 & 62 & 114 \\ 
  $c_\theta^{\text{Cl},180}(\vtrans(u; 0.5,\kappa), |2v-1|)$ & 2 & 452 & 152 & 281 & 246 & 141 & 216 \\ \midrule
   $c_{\nu,a_1,a_2}^{|t|}(|2u-1|, |2v-1|)$ & 3 & 678 & 227 & 113 & 614 & 267 & 131 \\ 
    $c_{\theta,a_1,a_2}^{\text{J}}(|2u-1|, |2v-1|)$ & 3 & 760 & 271 & 114 & 656 & 325 & 146 \\ 
    $c_{\theta,a_1,a_2}^{\text{Cl},180}(|2u-1|, |2v-1|)$  & 3 & 743 & 361 & 240 & 752 & 388 & 252 \\ \midrule
   $c_{\nu,a_1,a_2}^{|t|}(\vtrans(u; 0.5,\kappa), |2v-1|)$ & 4 & 0 & 2 & 0 & 0 & 0 & 0 \\ 
   $c_{\theta,a_1,a_2}^{\text{J}}(\vtrans(u; 0.5,\kappa), |2v-1|)$ & 4 & 88 & 52 & 1 & 37 & 63 & 14 \\ 
   $c_{\theta,a_1,a_2}^{\text{Cl},180}(\vtrans(u; 0.5,\kappa), |2v-1|)$ & 4 & 110 & 153 & 132 & 188 & 138 & 126 \\ \bottomrule
    \end{tabular}
  \caption{Comparison of AIC values when various h-symmetric copulas are
    fitted to simulated data from
    ARCH(1) and GARCH(1,1) processes with symmetric innovations and
    leverage by the PML method. Each column shows the differences between AIC
    and the minimum value of AIC for that column; thus a 0
    indicates the model with lowest AIC.\label{tab:2}}
\end{table}

In this section we consider the ARCH(1) and GARCH(1,1) models with
symmetric innovations and leverage. According to
Proposition~\ref{prop:symmetry}, the first-order copula $C_1$ is
h-symmetric but not jointly symmetric.
According to Example~\ref{ex:arch-leverage}, the ARCH(1) process has a
copula density of the form $c(\vtrans_1(u),|2v-1|)$ for an asymmetric
v-transform $\vtrans_1$ with fulcrum at $\delta=0.5$. For the
GARCH(1,1) model we don't have the same theoretical result but we will
try using copulas of the same form. The parameterisation of the
ARCH(1) model with leverage is the same as in Figure~\ref{fig:2} while
the parameters for
the GARCH(1,1) are given in Table~\ref{tab:parameters}.

In Table~\ref{tab:2} we include a selection of the copulas from
Table~\ref{tab:1}, omitting the mixture copulas, since the
inverse-v-transformed copulas typically do better, and also omitting
the inverse-v-transformed Gumbel copula, since the other inverse-v-transformed copulas do
better. We add copulas of the form
$c(u,v) = c^*(\vtrans_1(u), |2v-1|)$  where $\vtrans_1$ is a member of the
family $\vtrans(u; \delta,\kappa)$ in~(\ref{eq:2_v}) with
$\delta=0.5$, so that only one extra unknown parameter $\kappa$
is introduced.


The first panel contains jointly symmetric and exchangeable copulas;
the second panel contains h-symmetric copulas based on an exchangeable
copula for $c^*$ but incorporating an asymmetric v-transform $\vtrans(u;0.5,\kappa)$ for the
first variable;
the third panel contains jointly symmetric non-exchangeable copulas
using the Khoudraji construction~\eqref{eq:6};
the final panel contains h-symmetric copulas based
on a non-exchangeable $c^*$ and incorporating an asymmetric v-transform
for the first variable.

The best results in all cases are obtained from copulas developed from the
absolute spherical t copula $C^{|t|}_\nu$ that are not jointly
symmetric and that incorporate the non-symmetric v-transform, as theory suggests. In 5 of 6
cases the inverse v-transformation of the Khoudraji non-exchangeable extension of $C^{|t|}_\nu$ is
selected while in one case (ARCH(1) with t4 innovations) the inverse
v-transformation of $C^{|t|}_\nu$ is selected. Models based on the Joe
copula also do well. 


\subsection{Simulation study of case of asymmetric innovations}\label{sec:simul-study-asym}

\begin{table}[htb]
\centering
\begin{tabular}{llrrrr} \toprule
     & Model & \multicolumn{2}{c|}{ARCH(1)} &
                                             \multicolumn{2}{c}{GARCH(1,1)} \\
     \cmidrule(lr){3-4} \cmidrule(lr){5-6}
 Copula (form of density) & $p$ \textbar\ leverage & absent & present & absent & present \\
  \midrule
$ c^{|t|}_\nu(|2u-1|, |2v-1|) $ & 1 & 169 & 679 & 49 & 436 \\ 
   $c_\theta^{\text{J}}(|2u-1|, |2v-1|)$ & 1 & 221 & 728 & 92 & 487 \\ 
   $c_\theta^{\text{Cl},180}(|2u-1|, |2v-1|)$ & 1 & 377 & 881 & 208 & 673 \\ \midrule
  $ c^{|t|}_\nu(\vtrans(u; \delta_1), \vtrans(v;\delta_2))$  & 3 & 117 & 201 & 23 & 92 \\ 
   $ c^{\text{J}}_\nu(\vtrans(u; \delta_1), \vtrans(v;\delta_2))$  & 3 & 165 & 231 & 64 & 140 \\ 
    $ c^{\text{Cl},180}_\nu(\vtrans(u; \delta_1), \vtrans(v;\delta_2))$ & 3 & 348 & 558 & 189 & 403 \\ \midrule
   $ c^{|t|}_{\nu,a_1,a_2}(\vtrans(u; \delta_1), \vtrans(v;\delta_2))$ & 5 & 116 & 205 & 25 & 94 \\ 
   $ c^{\text{J}}_{\nu,a_1,a_2}(\vtrans(u; \delta_1), \vtrans(v;\delta_2))$ & 5 & 169 & 235 & 68 & 144 \\ 
   $ c^{\text{Cl},180}_{\nu,a_1,a_2}(\vtrans(u; \delta_1), \vtrans(v;\delta_2))$  & 5 & 337 & 561 & 191 & 403 \\ \midrule
   $ c^{|t|}_\nu(\vtrans(u; \delta_1,\kappa_1), \vtrans(v;\delta_2,\kappa_2))$  & 5 & 1 & 0 & 0 & 0 \\ 
    $ c^{\text{J}}_\nu(\vtrans(u; \delta_1,\kappa_1), \vtrans(v;\delta_2,\kappa_2))$ & 5 & 48 & 24 & 45 & 50 \\ 
    $ c^{\text{Cl},180}_\nu(\vtrans(u; \delta_1,\kappa_1), \vtrans(v;\delta_2,\kappa_2))$ & 5 & 229 & 293 & 159 & 277 \\ \midrule
   $ c^{|t|}_{\nu,a_1,a_2}(\vtrans(u; \delta_1,\kappa_1), \vtrans(v;\delta_2,\kappa_2))$  & 7 & 0 & 4 & 1 & 4 \\ 
    $ c^{\text{J}}_{\nu,a_1,a_2}(\vtrans(u; \delta_1,\kappa_1), \vtrans(v;\delta_2,\kappa_2))$  & 7 & 51 & 28 & 49 & 54 \\ 
    $ c^{\text{Cl},180}_{\nu,a_1,a_2}(\vtrans(u; \delta_1,\kappa_1), \vtrans(v;\delta_2,\kappa_2))$  & 7 & 222 & 297 & 160 & 279 \\ \bottomrule
\end{tabular}
 \caption{Comparison of AIC values when various copulas are
    fitted to simulated data from
    ARCH(1) and GARCH(1,1) processes with asymmetric innovations by the PML method. Each column shows the differences between AIC
    and the minimum value of AIC for that column; thus a 0
    indicates the model with lowest AIC.\label{tab:3}}
\end{table}

In the final experiment we consider models with skewed t innovations
with parameters $\nu=4$ and $\lambda = 0.8$. We simulate ARCH(1) and
GARCH(1,1) processes, both with and without leverage effects; parameter
values are as in Table~\ref{tab:parameters}.
We have already observed that asymmetric innovations remove all of the
previous symmetries we have discussed and so our approach in this
section is pure experimentation without the theoretical insights we
exploited in the previous experiments.

The first set of 3 models in Table~\ref{tab:3} consists of jointly
symmetric and exchangeable copulas, included as a benchmark. The
spherical t 
copula is the best of these but substantial improvements can be made
by introducing asymmetry.

The second set of 3 models is based on copula densities of the form
$c(u,v) = c^*(\vtrans(u; \delta_1), \vtrans(v; \delta_2))$ where
$\vtrans(\cdot;\delta)$ is the linear v-transform with fulcrum at
$\delta$ (asymmetric when $\delta \neq 0.5$) and $c^*$ is either the ast,
Clayton survival or Joe copula. It should be noted that full
  maximum likelihood estimation with respect to all parameters
  including the parameters of the v-transform presents no problems for
  these models since
the ast, Clayton survival and Joe copula densities do not have
asymptotes at $(0,0)$. If this were not the case we can run into a
threshold estimation issue, as discussed
in~\cite{bib:bladt-mcneil-21}. Adding the parameters $\delta_1$ and $\delta_2$ lowers the AIC values
considerably. The third set of 3
models applies the Khoudraji transformation to $c^*$ before
application of the inverse linear v-transform, but this barely
offers any improvement on the previous set of models; the AIC is only
really lowered for the Clayton survival copula in the ARCH(1) model
without leverage, and survival Clayton is the worst of
the three candidates in this case.

The fourth set of models is based on copula densities 
$c(u,v) = c^*(\vtrans(u;\delta_1,\kappa_1), \vtrans(v;\delta_2,
\kappa_2))$
where $\vtrans(\cdot;\delta, \kappa)$ is the 2-parameter v-transform
in~(\ref{eq:2_v}); such v-transforms are both non-linear and
asymmetric. These models give the best results of all, with the
model where $c^*$ is the ast copula best for 3 out of 4 datasets
datasets, followed by the
model based on Joe. Addition of the Khoudraji transformation in the
final set of models only offers an improvement for the ARCH(1) model
without leverage, and this improvement is very minor.

We conclude that allowing asymmetric v-transform largely removes
the need to consider the Khoudraji transformation, even when the
v-transforms are linear. This has positive consequences for practical
work: evaluating the density of Khoudraji-extended copulas of
type~\eqref{eq:6} is computationally challenging since the density
involves the copula $C_\theta$ and its h-functions, as well as is density
$c_\theta$. This is particularly burdensome when $C_\theta =
C^{|t|}_\nu$, the ast copula, as
numerical integration is required to evaluate $C^{|t|}_\nu$.

\section{D-vine processes using pair copulas based on linear
  v-transforms}\label{sec:d-vine-processes}

\subsection{Calculating $h$-functions}\label{sec:calc-h-funct}

The following results shows that all the formulas required to work
with linear v-transforms are easily obtained.

\begin{proposition}\label{prop:properties-linear-inv-v-transform}
  The copula $C$ with density given by $c(u,v) = c^*(\vtrans(u; \delta_1), \vtrans(v;\delta_2))$ is
\begin{equation}\label{eq:13}
C(u,v) = \delta_1^{\indicator{u \leq \delta_1}}(\delta_1-1)^{\indicator{u > \delta_1}}\delta_2^{\indicator{v \leq \delta_2}}(\delta_2-1)^{\indicator{v > \delta_2}} C^*(\vtrans(u; \delta_1),\vtrans(v; \delta_2)) + \delta_1v + \delta_2u -\delta_1\delta_2
\end{equation}
where $C^*$ is the copula with density $c^*$. Writing $h_1^*$ and
$h_2^*$ for the partial
derivatives of $C^*$, the corresponding partial derivatives
$h_1$ and $h_2$ of $C$ are 
\begin{align*}
  h_1(u,v) &= -
  \delta_2^{\indicator{v \leq \delta_2}}(\delta_2-1)^{\indicator{v >
      \delta_2}} h_1^*(\vtrans(u; \delta_1),\vtrans(v; \delta_2)) +
  \delta_2 \\
   h_2(u,v) &= -
  \delta_1^{\indicator{u \leq \delta_1}}(\delta_1-1)^{\indicator{u >
      \delta_1}} h_2^*(\vtrans(u; \delta_1),\vtrans(v; \delta_2)) +
              \delta_1 
\end{align*}
and satisfy the
equations
\begin{equation}\label{eq:14}
  \begin{aligned}
    \vtrans\left(h_1(u,v) ;\delta_2\right) & = 
                                                   h_1^*(\vtrans(u; \delta_1),\vtrans(v; \delta_2))
    \\
       \vtrans\left(h_2(u,v) ;\delta_1\right) & = 
                                                   h_2^*(\vtrans(u; \delta_1),\vtrans(v; \delta_2)).
  \end{aligned}
\end{equation}
Moreover the
inverses $h_1^{-1}(u,v) = \{x : h_1(u,x) =v\}$ and $h_2^{-1}(u,v) =
\{y : h_2(y,v) =u\}$ are 
\begin{align*}
  h_1^{-1}(u,v) &= -
  \delta_2^{\indicator{v \leq \delta_2}}(\delta_2-1)^{\indicator{v >
      \delta_2}} h_1^{*-1}(\vtrans(u; \delta_1),\vtrans(v; \delta_2)) +
                  \delta_2 \\
  h_2^{-1}(u,v) &= -
  \delta_1^{\indicator{u \leq \delta_1}}(\delta_1-1)^{\indicator{u >
      \delta_1}} h_2^{*-1}(\vtrans(u; \delta_1),\vtrans(v; \delta_2)) +
  \delta_1
\end{align*}
where $h_1^{*-1}$ and $h_2^{*-1}$ are the inverses of $h_1^*$ and $h_2^*$.
\end{proposition}
\begin{proof}
The joint density $c$ of the copula $C$ defined in ~\eqref{eq:13} can be easily shown to
satisfy $c(u,v) = \frac{\partial^2}{\partial u \partial v}C(u,v)  =
c^*(\vtrans(u; \delta_1), \vtrans(v;\delta_2))$. We just need to check
that it is a valid copula, which entails checking that
$C(0,v) = 0$ and $C(1,v) = v$ for any $v$, which is straightforward. By
symmetry, it then follows that $C(u,0) = 0$ and $C(u,1) = u$ for any
$u$.
The
$h$-functions are obtained by differentiation of~\eqref{eq:13} and the
identities~\eqref{eq:14} are easy to verify.
Noting that $h_1(u,\delta_2) = \delta_2$ and $h_2(u,\delta_1) = \delta_1$, it is also straightforward to
derive the inverses of the $h$-functions.
\end{proof}

 \subsection{Link to vt-D-vine models}\label{sec:link-vt-d}

In the special case of a model with two identical linear
v-transforms we now show that the resulting
 D-vine copula also belongs to the class of vt-D-vines
 in~\cite{bib:bladt-mcneil-21} which has been applied successfully to real
 financial data. The use of two different linear v-transforms does
 however give much more flexibility.

 For this result we
 introduce sequences of conditional distribution functions (Rosenblatt functions)
 $\Rbackward_k:(0,1) \times (0,1)^k \to (0,1)$ and
 $\Rforward_k:(0,1)\times(0,1)^k \to (0,1)$
 defined by
\begin{equation} \label{eq:Rosenblatt1}
   \begin{aligned}
   \Rbackward_{k-1}(x; \bm{u}) &= \P(U_1 \leq x \mid U_2
   =u_1,\ldots,U_{k}=u_{k-1}) = F_{U_1 \mid \bm{U}_{[2:k]}}(x \mid \bm{u}_{[1:k-1]}),\\
      \Rforward_{k-1}(x; \bm{u}) &=\P(U_{k} \leq x \mid U_{k-1}
      =u_1,\ldots,U_1 = u_{k-1}) = F_{U_{k} \mid \bm{U}_{[1:k-1]}}(x \mid \bm{u}_{[k-1:1]}),
    \end{aligned}
      \quad k \in \N.
    \end{equation}
  Note the reversal of the order of the conditioning values $\bm{u}_{[k-1:1]}
    =(u_{k-1}, \ldots,u_1)^\top$ in the second equation. Note also
    that we interpret
 $\Rbackward_0(x;\cdot) = \Rforward_0(x;\cdot) = x$.
 The Rosenblatt functions satisfy
$\Rforward_1(x;u) =
h_{1,1}(u,x)$, $\Rbackward_1(x;v) = h_{1,2}(x,v)$ and 
\begin{equation}\label{eq:16}
\begin{aligned}
   \Rforward_k(x;\bm{u})  & =  h_{k,1}\left(  \Rbackward_{k-1}(u_k;\bm{u}_{[k-1:1]})  ,
                         \Rforward_{k-1}(x; \bm{u}_{[1:k-1]})   \right), \\
    \Rbackward_k(x; \bm{u})  & =  h_{k,2}\left(
      \Rbackward_{k-1}(x; \bm{u}_{[1,k-1]})  ,
                         \Rforward_{k-1}(u_k; \bm{u}_{[k-1:1]})   \right),
                            \end{aligned}\quad k \geq 2,
\end{equation}
where $h_{k,1}(u,v)  = \frac{\partial}{\partial
  u}C_k(u,v)$ and $h_{k,2}(u,v)  = \frac{\partial}{\partial
  v}C_k(u,v)$ are the $h$-functions of the copula $C_k$ with
density $c_k$. This allows us to write the density in~\eqref{eq:15} as
 \begin{equation}\label{eq:155}
    c_{U_1,\ldots,U_n}(u_1,\ldots,u_n)  = \prod_{k=1}^{n-1} \prod_{j=k+1}^n c_k
  \Big(\Rbackward_{k-1}(u_{j-k}; \bm{u}_{[j-k+1:j-1]}),
    \Rforward_{k-1}(u_j; \bm{u}_{[j-1:j-k+1]}) \Big).
  \end{equation}


\begin{theorem}\label{theorem:vt-d-vine}
  Suppose we construct a d-vine copula~\eqref{eq:155}
  using a sequence of copulas $(C_k)_{k\in\N}$ whose densities are all
  of the form
  $c_k(u,v) = c_k^*(\vtrans(u;\delta),\vtrans(v;\delta))$ for
  a sequence of copulas $(C_k^*)_{k\in\N}$ with densities $(c_k^*)_{k\in\N}$ and a
  single linear v-transform $\vtrans(\cdot, \delta)$. Then the joint density
  $c_{U_1,\ldots,U_n}(u_1,\ldots,u_n)$ may be
  written as
  \begin{equation}\label{eq:10}
 \prod_{k=1}^{n-1} \prod_{j=k+1}^n c_k^*
  \Big(\Rbackward^*_{k-1}(\vtrans(u_{j-k}; \delta); \vtrans(\bm{u}_{[j-k+1:j-1]};\delta)),
    \Rforward^*_{k-1}(\vtrans(u_j; \delta);
    \vtrans(\bm{u}_{[j-1:j-k+1]}; \delta)) \Big)
  \end{equation}
  where $(\Rforward_k^*)_{k\in\N}$ and $(\Rbackward_k^*)_{k\in\N}$ are the
  forward and backward Rosenblatt functions for the
  copula sequence $(C_k^*)_{k\in\N}$ and where we use the compact notation $\vtrans(\bm{u};\delta) =
 (\vtrans(u_1;\delta),\ldots,\vtrans(u_d;\delta))^\top$.
\end{theorem}
\begin{proof}
  The result will follow if we can show that the equalities
  \begin{equation}\label{eq:888}
    \vtrans\left(    \Rbackward_{k-1}(u;
      \bm{v}) ;\delta \right) = \Rbackward^*_{k-1}(\vtrans(u; \delta);
    \vtrans(\bm{v};\delta))\quad \text{and} \quad
       \vtrans\left(    \Rforward_{k-1}(u;
      \bm{v}) ;\delta \right) = \Rforward^*_{k-1}(\vtrans(u; \delta);
    \vtrans(\bm{v};\delta))
  \end{equation}
  hold for any $k \in \N$, $u \in(0,1)$ and $\bm{v} \in (0,1)^{k-1}$. We
  will argue by induction. The equalities~\eqref{eq:888}
  are obviously true for $k=1$ since, by our convention, $B_0(u;\cdot) =
  B_0^*(u; \cdot) = u$  and $R_0(u;\cdot) =
  R_0^*(u; \cdot) = u$ for all $u$. Assume the equalities~\eqref{eq:888}
  are true when $k = l \in \N$ and
  consider $k = l+1$. For the first equality we can show, by first
  using formula~\eqref{eq:14} in Proposition~\ref{prop:properties-linear-inv-v-transform} and then applying the inductive assumption, that
  \begin{align*}
      \vtrans\left(    \Rbackward_{l}(u;
        \bm{v}) ;\delta \right) &= \vtrans\left(
      h_{l,2}\left(
      \Rbackward_{l-1}(u; \bm{v}_{[1,k-1]})  ,
      \Rforward_{l-1}(v_k; \bm{v}_{[k-1:1]})
                           \right);\delta \right) \\
    &= 
      h_{l,2}^*\left(
      \vtrans\left(\Rbackward_{l-1}(u; \bm{v}_{[1,k-1]}) ;\delta \right),
      \vtrans \left(\Rforward_{l-1}(v_k; \bm{v}_{[k-1:1]});\delta \right)
      \right) \\
                         &=
                           h_{l,2}^*\left(\Rbackward^*_{l-1}(\vtrans(u;
                           \delta); \vtrans(\bm{v}_{[1,k-1]};\delta)) ,
    \Rforward^*_{l-1}(\vtrans(v_k;\delta); \vtrans(\bm{v}_{[k-1:1]};\delta)) \right)
     = \Rbackward^*_{l}(\vtrans(u; \delta);
    \vtrans(\bm{v}; \delta))                    
    \end{align*}
    and the same argument can be applied to the second equality.
  \end{proof}

  Equation~(\ref{eq:10}) describes a vt-D-vine copula
  process. In such a model the copula density can be evaluated by
  first applying the v-transform to the entire vector
  $(u_1,\ldots,u_n)$ and then passing the resulting values to the
  density of a d-vine copula based on the sequence $(C_k^*)_{k\in\N}$.
  \cite{bib:mcneil-20} shows that if we apply such a copula process
  to a series $(X_t)_{t\in\Z}$ we effectively use the copula sequence
  $(C_k^*)_{k\in\N}$ to model the serial
  dependence structure of the time series
  $(T(X_t))_{t\in\Z}$ where $T$ is a v-shaped transformation on
  $\R$.

  Note that the Theorem above is a little more general than is
  needed. The copula sequence in our application is derived from a single
  parametric copula family $C^*$ and we have $C_k^*(u,v) =
  C^*(u,v;\theta_k)$ for a sequence of parameters
  $(\theta_k)_{k\in\N}$. However, we could obviously choose the
  $C_k^*$ from different parametric families.

  \section{Detailed estimation results}\label{sec:detailed-results}

  Table~\ref{tab:GARCHfit-symmetric} contains the estimation results
  for various D-vine models fitted to 10000 simulated data from a GARCH(1,1)
  process with symmetric innovations from a t distribution with 4
  degrees of freedom. The most effective models in
  terms of AIC and BIC are the ast-AR(5) and ast-ARMA(1,1) models.

  \begin{table}[ht]
\centering
\begin{tabular}{rrrrrr}
  \toprule
  Model & D-vine order & No.\ pars & Loglik & AIC & BIC \\ \midrule
T1 &  1 &  2 & 442.8 & -881.6 & -867.2 \\ 
  B1 &  1 &  5 & 436.6 & -863.2 & -827.1 \\ 
  Joe-AR(1) &  1 &  3 & 440.6 & -875.1 & -853.5 \\ 
  Clayton180-AR(1) &  1 &  3 & 424.6 & -843.2 & -821.6 \\ 
  ast-AR(1) &  1 &  3 & 442.8 & -879.5 & -857.9 \\ \midrule
  T2 &  2 &  4 & 576.3 & -1144.6 & -1115.7 \\ 
  B2 &  2 & 10 & 568.6 & -1117.3 & -1045.2 \\ 
  Joe-AR(2) &  2 &  4 & 564.4 & -1120.7 & -1091.9 \\ 
  Clayton180-AR(2) &  2 &  4 & 552.1 & -1096.2 & -1067.4 \\ 
  ast-AR(2) &  2 &  4 & 573.1 & -1138.2 & -1109.4 \\ \midrule
  T5 &  5 & 10 & 695.6 & -1371.3 & -1299.2 \\ 
  B5 &  5 & 25 & 689.7 & -1329.4 & -1149.2 \\ 
  Joe-AR(5) &  5 &  7 & 681.9 & -1349.8 & -1299.3 \\ 
  Clayton180-AR(5) &  5 &  7 & 667.9 & -1321.9 & -1271.4 \\ 
  ast-AR(5) &  5 &  7 & 693.9 & -1373.7 & -1323.3 \\ \midrule
  Joe-ARMA(1,1) & 20 &  4 & 682.0 & -1356.0 & -1327.2 \\ 
  Clayton180-ARMA(1,1) & 20 &  4 & 669.5 & -1331.0 & -1302.2 \\ 
  ast-ARMA(1,1) & 20 &  4 & 694.0 & -1380.0 & -1351.1 \\ 
   \bottomrule
\end{tabular}
\caption{Comparison of AIC and BIC values when various D-vines are
    fitted to simulated data from
    a GARCH(1,1) process with symmetric t4 innovations.\label{tab:GARCHfit-symmetric}}
\end{table}

Table~\ref{tab:GARCHfit-asymmetric} contains corresponding results for 10000 data from a
GARCH(1,1) process with skewed t4 innovations while
Table~\ref{tab:GARCHfit-asymmetric-leverage} is for a model which has
both skewed t4 innovations and leverage. Recall that parameter values
may be found in Table~\ref{tab:parameters}.
In both these analyses the preferred models are ast-AR(5) and
ast-ARMA(1,1). The improvement of ast-ARMA(1,1) on ast-AR(5) is not
quite so large as in the case of the empirical data in
Table~\ref{table:empirical},
since the persistence of volatility in the simulated GARCH processes
is lower.

  \begin{table}[ht]
\centering
\begin{tabular}{rrrrrr}
  \toprule
  Model & D-vine order & No.\ pars & Loglik & AIC & BIC \\ 
  \midrule
T1 &  1 &  2 & 503.0 & -1002.1 & -987.7 \\ 
  B1 &  1 &  5 & 499.5 & -989.0 & -953.0 \\ 
  Joe-AR(1) &  1 &  3 & 500.8 & -995.6 & -974.0 \\ 
  Clayton180-AR(1) &  1 &  3 & 481.1 & -956.2 & -934.6 \\ 
  ast-AR(1) &  1 &  3 & 505.5 & -1005.1 & -983.4 \\ \midrule
  T2 &  2 &  4 & 649.2 & -1290.5 & -1261.7 \\ 
  B2 &  2 & 10 & 650.7 & -1281.4 & -1209.3 \\ 
  Joe-AR(2) &  2 &  4 & 648.5 & -1289.0 & -1260.1 \\ 
  Clayton180-AR(2) &  2 &  4 & 626.3 & -1244.6 & -1215.7 \\ 
  ast-AR(2) &  2 &  4 & 656.4 & -1304.9 & -1276.0 \\ \midrule
  T5 &  5 & 10 & 747.5 & -1475.0 & -1402.9 \\ 
  B5 &  5 & 25 & 756.4 & -1462.8 & -1282.6 \\ 
  Joe-AR(5) &  5 &  7 & 744.4 & -1474.7 & -1424.2 \\ 
  Clayton180-AR(5) &  5 &  7 & 723.0 & -1431.9 & -1381.4 \\ 
  ast-AR(5) &  5 &  7 & 758.4 & -1502.8 & -1452.4 \\ \midrule
  Joe-ARMA(1,1) & 20 &  4 & 744.5 & -1481.0 & -1452.2 \\ 
  Clayton180-ARMA(1,1) & 20 &  4 & 726.5 & -1445.0 & -1416.1 \\ 
  ast-ARMA(1,1) & 20 &  4 & 759.1 & -1510.3 & -1481.4 \\ 
   \bottomrule
\end{tabular}
\caption{Comparison of AIC and BIC values when various D-vines are
    fitted to simulated data from
    a GARCH(1,1) process with skew t4 innovations.\label{tab:GARCHfit-asymmetric}}
\end{table}

  \begin{table}[ht]
\centering
\begin{tabular}{lrrrrr}
  \toprule
Model & D-vine order & No.\ pars & Loglik & AIC & BIC \\ 
  \midrule
T1 &  1 &  2 & 569.2 & -1134.4 & -1120.0 \\ 
  B1 &  1 &  5 & 588.6 & -1167.3 & -1131.2 \\ 
  Joe-AR(1) &  1 &  3 & 590.2 & -1174.3 & -1152.7 \\ 
  Clayton180-AR(1) &  1 &  3 & 568.5 & -1130.9 & -1109.3 \\ 
  ast-AR(1) &  1 &  3 & 597.6 & -1189.1 & -1167.5 \\ \midrule
  T2 &  2 &  4 & 723.8 & -1439.5 & -1410.7 \\ 
  B2 &  2 & 10 & 760.5 & -1501.0 & -1428.8 \\ 
  Joe-AR(2) &  2 &  4 & 754.4 & -1500.8 & -1472.0 \\ 
  Clayton180-AR(2) &  2 &  4 & 725.6 & -1443.3 & -1414.4 \\ 
  ast-AR(2) &  2 &  4 & 763.6 & -1519.1 & -1490.3 \\ \midrule
  T5 &  5 & 10 & 813.8 & -1607.6 & -1535.5 \\ 
  B5 &  5 & 25 & 888.2 & -1726.5 & -1546.3 \\ 
  Joe-AR(5) &  5 &  7 & 868.1 & -1722.1 & -1671.6 \\ 
  Clayton180-AR(5) &  5 &  7 & 828.8 & -1643.6 & -1593.1 \\ 
  ast-AR(5) &  5 &  7 & 879.6 & -1745.1 & -1694.6 \\ \midrule
  Joe-ARMA(1,1 & 20 &  4 & 866.0 & -1724.0 & -1695.1 \\ 
  Clayton180-ARMA(1,1) & 20 &  4 & 831.0 & -1654.1 & -1625.2 \\ 
  ast-ARMA(1,1) & 20 &  4 & 880.0 & -1752.0 & -1723.2 \\ 
   \bottomrule
\end{tabular}
\caption{Comparison of AIC and BIC values when various D-vines are
    fitted to simulated data from
    a GARCH(1,1) process with skew t4 innovations and leverage.\label{tab:GARCHfit-asymmetric-leverage}}
\end{table}

\section{Analysis of VaR exceedances}\label{sec:analys-var-exce}

After model estimation, the Rosenblatt functions $R_k$ in~\eqref{eq:Rosenblatt1} can be used to construct
one-step-ahead conditional quantile or value-at-risk (VaR) forecasts for
all our fitted D-vine copula models. We follow \cite{bib:loaiza-maya-et-al-18} by using this
as the basis of a method for evaluating the quality of the fitted
models. Specifically, for $t=2,\ldots,n = 3668$ we
compare the realized values $u_t$ with VaR forecasts
$\operatorname{VaR}_t(\alpha)$ derived according to the formula
\begin{equation*}
  \operatorname{VaR}_t(\alpha) = R_{k(t)}^{-1}(\alpha ;
  u_{t-1},\ldots, u_{t - k(t)})
\end{equation*}
 for $\alpha \in \{0.01,0.05, 0.1, 0.9,
0.95, 0.99\}$. The number of conditioning variable in the forecast is
taken to be $k(t) = \min(t-1, p, 12)$ where $p$ is the Markov order of the
fitted model and 12 is a cap on the Markov order to ensure reasonably
fast run times. Note that, we do not
compare the realized raw data $x_t$ with estimates of the conditional
quantiles of the composite model for $(X_t)$, but rather we compare the data transformed to the copula scale $u_t =
\widehat{F}_X(x_t)$ with estimates of the conditional quantiles of the
D-vine process
$(U_t)$; this means that we evaluate the quality of the copula model
independently of the quality of the adaptive kernel density estimator
$\widehat{F}_X$ of~\cite{bib:loaiza-maya-et-al-18} which is used for
the marginal distribution $F_X$.

The hit rates
$\hat{\alpha} = \tfrac{1}{n-1}\sum_{t=2}^n \indicator{u_t <
  \operatorname{VaR}_t(\alpha)}$ for our favoured models (Joe-AR(5), ast-AR(5), Clayton180-ARMA(1,1) and
  ast-ARMA(1,1)) are shown in
Table~\ref{tab:var_results} and are very similar to the rates reported
by~\cite{bib:loaiza-maya-et-al-18}. In all cases the conditional
coverage test of~\cite{bib:christoffersen-98} leads to an insignificant
test result with confidence level 99\%.

\begin{table}[ht]
  \centering
  \begin{tabular}{lrrrrrr} \toprule
    & \multicolumn{6}{l}{Quantile level $\alpha$ (percent)} \\
    Model   &1 & 5&10&90&95& 99\\ \midrule
  Joe-AR(5) & 0.98 & 4.85 & 9.54 & 90.35 & 95.07 & 99.13 \\ 
  ast-AR(5) & 0.95 & 5.04 & 9.60 & 90.08 & 94.98 & 99.18 \\ 
  ast-ARMA(1,1) & 1.06 & 4.83 & 9.87 & 90.21 & 94.98 & 99.05 \\ 
  Clayton180-ARMA(1,1) & 1.04 & 4.88 & 10.33 & 90.19 & 94.77 & 99.15 \\ \bottomrule
  \end{tabular}
  \caption{Empirical hit rates (expressed as percentages) for
    1-day-ahead conditional $\alpha$-quantile (VaR)
    estimates in analysis of USD/AUD
    exchange-rate returns. For details of models see
    Section~\ref{sec:examples} and Table~\ref{table:empirical}.}
  \label{tab:var_results}
\end{table}

\section{Copula of an AR(1) model with ARCH(1) errors}\label{sec:ar1-model-with}

In this section we consider a time
series model with stochastic volatility as well as a non-zero
stochastic conditional mean term. Specifically we consider an AR(1)
model with ARCH(1) errors which is a
  first-order Markov
  process $(X_t)_{t\in\Z}$ 
satisfying equations
\begin{equation}\label{eq:9}
  X_t = \phi X_{t-1} +\volfunc(X_{t-1}) \;
  \epsilon_t\quad\text{with}\quad \volfunc(x) = \sqrt{\alpha_0 + \alpha_1 x^2},
\end{equation}
where the innovations $(\epsilon_t)_{t\in\Z}$ have a symmetric
distribution. This model was proposed and
analysed by~\cite{bib:borkovec-klueppelberg-01} who showed that it is covariance stationary
provided $\phi^2 + \alpha_1 < 1$ and has a symmetric
stationary distribution $F_X$. Note that it differs
  from the usual  version of the AR(1)-ARCH(1) model, $X_t = \phi
  X_{t-1} +\volfunc(X_{t-1}-\phi X_{t-2})\; \epsilon_t$,
  but this is not first-order Markov and thus less straightforward to treat
  with our methods.

The formulas~(\ref{eq:22}) require only simple modification to
account for the non-zero conditional mean term. In particular, the copula becomes
\begin{equation*}
   c_1(u,v) = \frac{1}{\volfunc\left(F_X^{-1}(u)\right)}f_\epsilon \left(
     \frac{F_X^{-1}(v) - \phi F_X^{-1}(u)}{\volfunc\left(F_X^{-1}(u)\right)} \right)
   \frac{1}{f_X(F_X^{-1}(v))} \,.
 \end{equation*}
It may be easily verified that the copula density satisfies $c(u,v) =
c(1-u,1-v)$ and so the copula is radially symmetric. However, the
copula is neither h-symmetric nor v-symmetric, and thus it is not
jointly symmetric.
The copula density is shown in the bottom right panel of Figure~\ref{fig:OS} in the
case where $\phi = 0.3$, $\alpha_0 = 0.4$ and $\alpha_1 = 0.6$. It is
evident that more of the probability mass is now on the main diagonal
of the copula. Behaviour of this kind could perhaps be modelled by a
mixture of a Gauss copula (or Frank copula) and one of the inverse-v-transformed copulas
examined earlier.

\bibliographystyle{apalike}
\setcitestyle{authoryear,open={(},close={)}}

\newcommand{\noopsort}[1]{}

\end{document}